\documentclass[12pt, draftclsnofoot, onecolumn]{IEEEtran}

\usepackage{amsmath, amssymb, cite}
\usepackage{mathtools}
\usepackage[small]{caption}
\usepackage{amsthm,multirow,color,amsfonts}
\usepackage{tabulary}
\usepackage{subfigure}
\usepackage{graphicx}
\usepackage{setspace}
\usepackage{enumerate}
\usepackage[]{algorithm2e}
\usepackage{comment}

\title{Spatially Correlated Content Caching for Device-to-Device Communications}

\author{Derya~Malak, Mazin Al-Shalash, and~Jeffrey~G.~Andrews
\thanks{This work in part appeared in Proc. IEEE Intl. Symposium on Info. Theory, Barcelona, Spain, July
2016 \cite{Malak2016}. }
\thanks{D. Malak and J. G. Andrews are with the Wireless Networking and Communications Group (WNCG), The University of Texas at Austin, Austin, TX 78701 USA (email: deryamalak@utexas.edu; jandrews@ece.utexas.edu). M. Al-Shalash is with Huawei Technologies, Plano, TX 75075 USA (e-mail: mshalash@huawei.com). Last revised: {\today}.}
}

\newtheorem{theo}{Theorem}
\newtheorem{defi}{Definition}

\newtheorem{prop}{Proposition}

\newtheorem{ex}{Example}

\DeclareMathOperator*{\argmax}{arg\,max}
\DeclareMathOperator*{\Rdd}{R_{\sf D2D}}
\DeclareMathOperator*{\Rdds}{R_{\sf D2D}^2}
\DeclareMathOperator*{\PX}{p_{c,\Pi}}
\DeclareMathOperator*{\PI}{p_{c,I}}
\DeclareMathOperator*{\PP}{p_{c,MPC}}%Most popular
\DeclareMathOperator*{\PG}{p_{c,G}} 
\DeclareMathOperator*{\PGstar}{p_{c,G}^*}
\DeclareMathOperator*{\PGL}{p^{\rm Lin}_{c,G}}
\DeclareMathOperator*{\PMA}{p_{c,\mhcA}}
\DeclareMathOperator*{\PMAstar}{p_{c,\mhcA}^*}

\DeclareMathOperator*{\PhitI}{P_{\sf Hit,I}}%Independent
\DeclareMathOperator*{\PhitG}{P_{\sf Hit,G}}%PPP
\DeclareMathOperator*{\PhitGstar}{P_{\sf Hit,G}^*}%PPP    
\DeclareMathOperator*{\PhitMP}{P_{\sf Hit,MPC}}%Most popular
\DeclareMathOperator*{\PhitMA}{P_{\sf Hit,\mhcA}} %\mhcA
\DeclareMathOperator*{\PhitMALB}{P_{\sf Hit,\mhcA}^{\sf LB}} %\mhcA LB
\DeclareMathOperator*{\PhitMALBstar}{P_{\sf Hit,\mhcA}^{\sf LB*}} %\mhcA LB
\DeclareMathOperator*{\PhitMAUB}{P_{\sf Hit,\mhcA}^{\sf UB}} %\mhcA UB
\DeclareMathOperator*{\PhitMB}{P_{\sf Hit,\mhcB}} %\mhcB
\DeclareMathOperator*{\PhitN}{P_{\sf Hit,N}} %negatively placement
\DeclareMathOperator*{\PhitX}{P_{{\sf Hit},\Pi}} %X placement
\DeclareMathOperator*{\PmissX}{P_{{\rm Miss},\Pi}} %X placement
\DeclareMathOperator*{\PmissI}{P_{\rm Miss,I}} %Independent
\DeclareMathOperator*{\PmissN}{P_{\rm Miss,N}} %Negative
\DeclareMathOperator*{\PmissMP}{P_{\rm Miss,MPC}} %Most popular
\DeclareMathOperator*{\PmissMA}{P_{\rm Miss,MA}} %\mhcA
\DeclareMathOperator*{\mc}{m_c}

\DeclareMathOperator*{\BM}{\sf BM}

\DeclareMathOperator*{\PPP}{\sf PPP}
\DeclareMathOperator*{\GPP}{\sf GPP}

\DeclareMathOperator*{\GCP}{\sf GCP}
\DeclareMathOperator*{\MPC}{\sf MPC}
\DeclareMathOperator*{\HCP}{\sf HCP}
\DeclareMathOperator*{\mhc}{\sf MHC}
\DeclareMathOperator*{\mhcA}{{\sf HCP}-A}
\DeclareMathOperator*{\mhcB}{{\sf HCP}-B}
\DeclareMathOperator*{\DD}{\sf D2D}
\DeclareMathOperator*{\LRU}{\sf LRU}

\DeclareMathOperator*{\Poisson}{\sf Poisson}

\DeclareMathOperator*{\RHS}{\sf RHS}

\DeclareMathOperator*{\tx}{\lambda_t}%density of active transmitters

\DeclareMathOperator*{\rx}{\lambda_r}

\DeclareMathOperator*{\txMA}{\lambda_{\mhcA}}
\DeclareMathOperator*{\txMAs}{\lambda_{\mhcA}^2}
\DeclareMathOperator*{\txMAstar}{\lambda_{\mhcA}^*}

\SetKwInput{KwInput}{Input}
\SetKwInput{KwInitialization}{Initialization}

\begin{document}
\maketitle
\begin{abstract}
We study optimal geographic content placement for device-to-device ($\DD$) networks in which each file's popularity follows the Zipf distribution. The locations of the $\DD$ users (caches) are modeled by a Poisson point process ($\PPP$) and have limited communication range and finite storage. Inspired by the Mat\'{e}rn hard-core (type II) point process that captures pairwise interactions between nodes, we devise a novel spatially correlated caching strategy called {\em hard-core placement} ($\HCP$) such that the $\DD$ nodes caching the same file are never closer to each other than the {\em exclusion radius}. The exclusion radius plays the role of a substitute for caching probability. We derive and optimize the exclusion radii to maximize the \emph{hit probability}, which is the probability that a given $\DD$ node can find a desired file at another node's cache within its communication range. Contrasting it with independent content placement, which is used in most prior work, %our analysis shows that 
our $\HCP$ strategy often yields a significantly higher cache hit probability. We further demonstrate that the $\HCP$ strategy is effective for small cache sizes and a small communication radius, which are likely conditions for $\DD$.
\end{abstract}

%\begin{IEEEkeywords}
%Space, time, Distributed caching, device-to-device communications, content distribution, stochastic geometry.
%\end{IEEEkeywords}

\maketitle

%%%%%%%%%%%%%%%%%%%%%%%%%%%%%%%%%%%%%%%%%%%
\section{Introduction}
\label{intro} 
$\DD$ communication is a promising technique for enabling proximity-based applications and increased offloading from the heavily loaded cellular network, and is being actively standardized by 3GPP \cite{LinMag2014}. The efficacy of $\DD$ caching networks relies on users possessing content that a nearby user wants. Therefore, intelligent caching of popular files is critical for $\DD$ to be successful. Caching has been shown to provide increased spectral reuse and throughput gain in $\DD$-enabled networks \cite{Naderializadeh2014}, and the optimal way to cache content is studied from different perspectives, e.g. using probabilistic placement \cite{Blaszczyszyn2014}, maximizing cache-aided spatial throughput \cite{ChenPapKoun2016}, but several aspects of optimal caching exploiting spatial correlations for network settings have not been explored. Intuitively, given a finite amount of storage at each node, popular content should be seeded into the network in a way that maximizes the \emph{hit probability} that a given $\DD$ device can find a desired file -- selected at random according to a request distribution -- within its radio range. We explore this problem quantitatively in this paper by considering different spatial content models and deriving, optimizing and comparing the hit probabilities for each of them.

Content caching has received significant attention as a means of improving the throughput and latency of networks without requiring additional bandwidth or other technological improvements. A practical use case is video, which will consume nearly $80\%$ of all wireless data by 2021 \cite{Cisco2017}. %Video caching appears particularly profitable and plausible compared to other types of content \cite{Tadrous2015}
Video caching is perfectly suited to $\DD$ networks for offloading traffic from %congested 
cellular networks.

%%%%%%%%%%%%%%%%%%%%%%%%%%%%%%%%%%%%%%%%%%%%%
\subsection{Related Work and Motivation}\label{relatedwork}
Research to date on content caching has been mainly focused on two different perspectives. On one hand, researchers have attempted to understand the fundamental limits of caching gain. The gain offered by local caching and broadcasting is characterized in the landmark paper \cite{MaddahAli2013Journal}. Although this work does not deal with $\DD$ communications and the caches cannot cooperate, it provides the first attempt to characterize the gain offered by local caching. Scaling of %laws for 
the number of active $\DD$ links and optimal collaboration distance with $\DD$ caching are studied in \cite{Ji2014}, \cite{Golrezaei2014}. Combining random independent caching with short-range $\DD$ communications can significantly improve the throughput \cite{JiCaiMol2015}. Capacity scaling laws in wireless ad hoc networks are investigated in \cite{Gupta2000TIT}, featuring short link distances, and cooperative schemes for order optimal throughput scaling is proposed in \cite{Niesen2009}. Capacity scaling laws for single \cite{JiCaiMol2015}, \cite{MaddahAli2013Journal} and multi-hop caching networks \cite{Jeon2015} are also investigated. Physical layer caching is studied in \cite{NadMadAve2016} to mitigate the interference, and in \cite{LiuLau2016tnet} to achieve linear capacity scaling. Finite-length analysis of random caching schemes that achieve multiplicative caching gain is presented in \cite{Shanmugam2016TIT}, \cite{Vettigli2015}.
%Link layer caching is studied in \cite{AfaAndSno2008} to provide an extension to the standard 802.11 RTS/CTS (Request to Send/Clear to Send) link-layer protocol to make use of overheard packets with a packet ID check. The proposed ID check can reduce the air time usage by 25\%. Network layer caching for multicasting and addressed various challenges in multicasting such as managing bandwidth utilization of bottleneck links, not overloading the sender with retransmission requests and keeping the latency of retransmission low \cite{NeeRam2002}. Reliable delivery of cached content in the transport layer as a function of the cache size and hop length is investigated in \cite{Tiglao2012}.

Alternatively, as in the current paper, there are several studies focusing on decentralized caching algorithms that have optimized the caching distribution to maximize the cache hit probability, using deterministic or random caching as in \cite{Golrezaei2014TWC}, \cite{Ji2014} given a base station (BS)-user topology. FemtoCaching replaces backhaul capacity with storage capacity at the small cell access points, i.e., helpers, and the optimum way of assigning files to the helpers is analyzed in \cite{Shanmugam2013} to minimize the delay. There are also geographic placement models focusing on finding the cache locally such as \cite{Blaszczyszyn2014}, in which the cache hit probability is maximized for SINR, Boolean and overlaid network coverage models, and \cite{Malak2016_D2DCaching}, in which the density of successful receptions is maximized using probabilistic placement. Although most of these strategies suggest that the caching distribution should be skewed towards the most popular content and exploit the diversity of content, and it is not usually optimal to cache just the most popular files, as pointed out in \cite{Ji2014}, \cite{Golrezaei2014}. Further, as the current paper will show, unlike the probabilistic policies, where the files are independently placed in the cache memories of different nodes according to the same distribution \cite{Blaszczyszyn2014}, \cite{EliBartek2017}, and \cite{Malak2016_D2DCaching}; it is not usually optimal to cache files independently. %The cache placement should not be oblivious to the configurations of neighboring nodes and optimal caching should exploit the spatial correlations of the nodes. 
For larger transmission range and higher network density, we will quantify and see that the hit-maximizing caching strategy can be increasingly skewed away from independently caching the %most 
popular files. %Our main objective is to generalize the $\GCP$ model to correlated placement of the files and evaluate the performance as a function of $\Rdd$ and the exclusion radii.

%Some of the existing work exploit the connection between optimizing caching and throughput. 
Recent studies also address problems at the intersection of the hit probability and the spatial throughput. The spatial throughput in $\DD$ networks is optimized by suitably adjusting the proportion of active devices in \cite{KeeBlaMuh2016}. Exploiting stochastic geometry, a Poisson cluster model is proposed in \cite{Afshang2016} and the area spectral efficiency is maximized assuming that the desired content is available inside the same cluster as the typical device. Some of the existing work focuses on mitigating excessive interference to maximize the throughput or capacity, as in \cite{LiuLau2016tnet}, \cite{Naderializadeh2014}, \cite{NadMadAve2016}. Employing probabilistic caching, cache-aided throughput, which measures the density of successfully served requests by local device caches, is investigated in \cite{ChenPapKoun2016}. The optimal caching probabilities obtained by cache-aided throughput optimization provide throughput gain, particularly in dense user environments compared with the cache-hit-optimal case.

\begin{comment}
A related line of work is proactive caching -- also known as ``pre-fetching" -- of high bandwidth content at the network edge \cite{Tadrous2015}.  Pre-fetching can smooth out the network traffic to avoid peaks in space and time. Proactive caching can occur at the end user \cite{Bastug2014mag} or at small cell BSs to reduce peak demands on the backhaul link, which is often the bottleneck \cite{Shanmugam2013}. Despite their usefulness, none of these approaches capture the spatial interactions in the network or provide guidance on how the cache states at nearby nodes should affect the caching strategy at a given node.
\end{comment}

Challenges for the adoption of caching for wireless access networks also include making timely estimates of varying content popularity  \cite{Leconte2016}. Cache update algorithms exploiting the temporal locality of the content have been well studied \cite{Che2002}. Inspired from the Least Recently Used ($\LRU$) replacement principle, a multi-coverage caching policy at the edge-nodes is proposed in \cite{Giovanidis2016}, where caches are updated in a way that provides content diversity to users who are covered by more than one node. Although \cite{Giovanidis2016} combines the temporal and spatial aspects of caching and approaches the performance of centralized policies, it is restricted to the $\LRU$ principle.  

%To the best of our knowledge, a decentralized content caching policy that captures the joint spatial interactions of the nodes has not been studied in the literature. We aim to maximize the cache hit probability for a $\DD$ network where the spatial distribution of nodes can be exploited to decide where to store files and efficiently use the caches. 

%%%%%%%%%%%%%%%%%%%%%%%%%%%%%%%%%%%%%%%%%%%%%
\subsection{Contributions and A High Level Summary}
\label{contributions}
We consider a spatial $\DD$ network setting in which the $\DD$ user locations are modeled by a Poisson point process ($\PPP$), and users have limited communication range and finite storage. The $\DD$ users are served by each other if  the desired content is cached at a user within its radio range: this is called a \emph{hit}.  Otherwise, they are served by the cellular network base station, which is what $\DD$ communication aims to avoid. %We wish to optimize the cache hit probability.
%  over conventional caching schemes (which exploit purely the cached-assisted ...).

We concentrate exclusively on the content placement phase in the above %$\DD$ network 
setting in order to maximize the cache hit probability via exploiting the spatial diversity. We do not focus on the %content 
transmission phase that incorporates the path loss, fading or interference. The coverage process of the proposed scheme is represented by a Boolean model (BM). The BM is tractable for the noise-limited regime \cite{Blaszczyszyn2014}, where the interference is small compared to the noise. The coverage area of the BM is determined by a \emph{fixed communication radius}, as will be detailed in Sect. \ref{hitprobability}. 

%KIND OF REPETITIVE (NEXT 3 PARAGRAPHS)
{\bf Spatial caching, pairwise interactions and Mat\'{e}rn hard-core-inspired placement.} We introduce a spatial content distribution model for a $\DD$ network, and describe the cache hit probability maximization problem in Sect. \ref{hitprobability}. Our aim is to extend the independent content placement strategy, also known as geographic content placement ($\GCP$) \cite{Blaszczyszyn2014}, where there is no spatial correlation in placement, which we discuss in Sect. \ref{cachemodel-independent}. Exploiting the Mat\'{e}rn hard-core ($\mhc$) models, we propose novel spatially correlated cache placement strategies that enable spatial diversity to maximize the $\DD$ cache hit probability. In Sect. \ref{hardcore}, we detail the $\mhc$ placement and analyze two different $\mhc$ placement strategies: (i) $\mhcA$ that can provide a significantly higher cache hit probability than the $\GCP$ scheme in the small cache size regime and (ii) $\mhcB$ that has a higher hit probability than $\GCP$ for short ranges. 

{\bf The key differences from the independent placement model.} The device locations follow the $\PPP$ distribution, which provides a random deployment instead of a fixed pattern, and hence it is possible to have cache clusters and isolated caches \cite{Andrews2011}, and the content placement distribution is optimized accordingly. Unlike the independent placement model, where the cache placement distribution is independent and identically distributed (i.i.d.) over the spatial domain, the $\mhc$ model captures the pairwise interactions between the $\DD$ nodes and yields a negatively correlated placement. The caches storing a particular file are never closer to each other than some given distance, called the exclusion radius, meaning that neighboring users are not likely to cache redundant content. Hence, the radius of exclusion plays the role of a substitute for caching probability.

{\bf Comparisons and design insights.} Sect. \ref{comparisonindependentmhc} provides a simulation study to compare the performance between the different content placement strategies. Independent content placement does not exploit $\DD$ interactions at the network level, and our results show that geographic placement should exploit locality of content, which is possible through negatively correlated placement. For short range communication and small cache sizes, $\HCP$ is preferred, and when the network intensity is fixed, the cache hit rate gain of the $\HCP$ model over the $\GCP$ and caching most popular content schemes can reach up to $37\%$ and $50\%$, respectively when the communication range is improved, as demonstrated in Sect. \ref{comparisonindependentmhc}. %for long range communication and large cache sizes, independent placement model is a better alternative.

%%%%%%%%%%%%%%%%%%%%%%%%%%%%%%%%%%%%%%%%%%%
\section{System Model and Problem Formulation}
\label{hitprobability}
The locations of the $\DD$ users are modeled by a $\PPP$ $\Phi$ with density $\tx$ as in \cite{Lin2013}. We assume that there are $M$ total files in the network, where all files have the same size, and each user has the same cache size $N<M$. Depending on its cache state, each user makes requests for new files based on a general popularity distribution over the set of the files. The popularity of such requests is modeled by the Zipf distribution, which has probability mass function (pmf) $p_r(n)=\frac{1}{n^{\gamma_r}}/\sum_{m=1}^M{\frac{1}{m^{\gamma_r}}}$, for $n=1,\hdots, M$, where $\gamma_r$ is the Zipf exponent that determines the skewness of the distribution. The demand profile is Independent Reference Model (IRM), i.e., the standard synthetic traffic model in which the request distribution does not change over time \cite{Traverso2013}. Our objective is to maximize the average cache hit probability performance of the proposed caching model. Therefore, it is sufficient to consider a snapshot of the network\footnote{Extension of the model to also incorporate the temporal correlation of real traffic traces can be done by exploiting models like the Shot-Noise Model (SNM). This overcomes the limitations of the IRM by explicitly accounting for the temporal locality in requests for contents \cite{Traverso2013}. However, in that case, the problem under study will have an additional dimension to optimize over, and to do so, online learning algorithms should be developed to both learn the demand and optimize the spatial placement. The study of the temporal dynamics of the request distribution and the content transmission phase is left as future work.}, in which the $\DD$ user realization is given and requests are i.i.d. over the space. We devise a spatially correlated probabilistic placement policy, in which the $\DD$ caches are loaded in a distributed manner via additional marks attached to them without accounting for any cost, in a timescale that is much shorter than the time over which the %device 
locations are predicted, as will be detailed in Sect. \ref{hardcore}.
  
Consider a given realization $\phi=\{x_i\}\subset \mathbb{R}^2$ of the $\PPP$ $\Phi$. The coverage process of the proposed model can be represented by a Boolean model (BM) \cite[Ch. 3]{BaccelliBook1}. Specifically, given a transmit power $P$, if we only consider path loss (with exponent $\alpha$), no fading and no interference, the received signal at the boundary should be larger than a threshold to guarantee coverage, i.e., $Pr^{-\alpha}\geq T$, yielding $r\leq \Rdd=(P/T)^{1/\alpha}$. %Hence, the area of coverage is determined by a fixed communication radius, denoted by $\Rdd$.
Hence, $\DD$ users can only communicate within a finite range, which we call the $\DD$ radius, denoted by $\Rdd$. A file request is fulfilled by the $\DD$ users within $\Rdd$ if one has the file; else the $\DD$ user is served by a BS.

The BM is driven by %the independently marked 
a $\PPP$ on $\mathbb{R}^2$ in which the marks are deterministic (constant), $\tilde{\Phi}=\sum\nolimits_{i}{\delta_{(x_i,B_i(\Rdd))}}$, whose points $x_i$'s denote the germs, and on disc-shaped grains $B_i(\Rdd)$ -- a closed ball of fixed radius $\Rdd$ centered at $x_i$ -- that model the coverage regions of germs. The $\BM$ is a tractable model for the noise-limited regime \cite{Blaszczyszyn2014}. The coverage process of the $\DD$ transmitters driven by the $\BM$ is given by the union $V_{\rm BM}=\bigcup_i{(x_i+B_0(\Rdd))}$ \cite[Ch. 3]{BaccelliBook1}. {For the interference-limited regime, there is no notion of communication radius, and the analysis of the coverage becomes more involved. SINR coverage models as in \cite{Blaszczyszyn2014} can be exploited to determine the distribution of the coverage number, i.e., the number of $\DD$ users covering the typical receiver. However, this is beyond the  scope of the current paper.} %Therefore, our work is applicable only to the interference-free, i.e., noise-limited, case. 

To characterize the successful transmission probability, one needs to know the number of users that a typical node can connect to, i.e., the coverage number. Exploiting the properties of the $\PPP$, the distribution of the number of transmitters covering the typical receiver is given by $\mathcal{N}_P \sim\Poisson(\tx \pi {\Rdds})$. Therefore, 
\begin{align}
\mathbb{P}(\mathcal{N}_P=k)=e^{-\tx\pi{\Rdds}}\frac{(\tx\pi{\Rdds})^k}{k!},\quad k\geq 0.
\end{align}

%%%%%
\subsection{Cache Hit Probability}
\label{cachehitprobability}
Assume that the cache placement at the $\DD$ users is done in a dependent manner. Given $\mathcal{N}_P=k$ transmitters cover the typical receiver, let $Y_{(m,i)}$ be the indicator random variable that takes the value $1$ if file $m$ is available in the cache located at $x_i\in\phi$ and $0$ otherwise. Thus, the caching probability of file $m$ in cache $i$  is given by $\PX(m,x_i)=\mathbb{P}(Y_{(m,i)}=1)$. Optimal content placement is a binary problem where the cache placement constraint $\sum\nolimits_{m=1}^M{Y_{(m,i)}}\leq N$ is satisfied for all $x_i\in\phi$, i.e., $Y_{(m,i)}$'s are inherently dependent. However, the original problem is combinatorial and is NP-hard \cite{Shanmugam2013}. For tractability reasons, we take the expectation of this relation and obtain our relaxed cache placement constraint: $\sum\nolimits_{m=1}^M{\PX(m,x_i)}\leq N$. Later, we show there are feasible solutions to the relaxed problem filling up all the cache slots.

The maximum average total cache hit probability, i.e., the probability that the typical user finds the content in one of the $\DD$ users it is covered by, for a content placement strategy $\Pi$ can be evaluated by solving the following optimization formulation:
\begin{eqnarray}
\begin{aligned}
\label{eq:hitprobX}
\max_{\PX} &\,\,\, \PhitX\\
\textrm{s.t.}
& \quad \sum\limits_{m=1}^M{\PX(m,x_i)}\leq N,\quad x_i\in \Phi,
\end{aligned}
\end{eqnarray}
where the hit probability is given by the following expression:
\begin{align}
\label{PhitX}
\PhitX=1-\sum\limits_{m=1}^M{p_r(m)\sum\limits_{k=0}^{\infty}{\mathbb{P}(\mathcal{N}_{\Pi}=k)\PmissX(m,k)}}, 
\end{align}
where $\mathbb{P}(\mathcal{N}_{\Pi}=k)$ is the probability that $k$ transmitters (caches) cover the typical receiver, and $\PmissX(m,k)$ is the probability that $k$ caches cover a receiver, and none has file $m$.  

We propose different strategies to serve the $\DD$ requests that maximize the cache hit probability. Assuming a transmitter receives one request at a time and multiple transmitters can potentially serve a request, the selection of an active transmitter depends on the caching strategy. A summary of the symbol definitions and important network parameters are given in Table \ref{table:tab1}.

%%%%%
\subsection{Repulsive Content Placement Design}
\label{negdepplace}
Optimizing the marginal distribution for content caching by decoupling the caches of $\DD$ users in a spatial network scenario is not sufficient to optimize the joint performance of the caching. The performance can be improved by developing spatially correlated content placement strategies that exploit the spatial distribution of the $\DD$ nodes, as we propose in this paper. 

Negatively correlated spatial placement corresponds to a distance-dependent thinning of the transmitter process so that neighboring users are less likely to have matching contents. This kind of approach is promising from an average cache hit rate optimization perspective. Therefore, we mainly focus on negatively dependent or repulsive content placement strategies. 

We next define negative dependence for a collection of random variables.
\begin{defi}
\label{NegativeDependence}
Random variables $Y_1, \hdots, Y_k$, $k\geq 2$, are said to be negatively dependent, if for any numbers $y_1,\hdots, y_k\in\mathbb{R}$, we have that \cite{Gerasimov2012}
\begin{align}
\mathbb{P}\Big(\bigcap\nolimits_{i=1}^k {Y_i\leq y_i}\Big)\leq \prod\nolimits_{i=1}^k{\mathbb{P}(Y_i\leq y_i)},\quad%\nonumber\\
\mathbb{P}\Big(\bigcap\nolimits_{i=1}^k {Y_i> y_i}\Big)\leq \prod\nolimits_{i=1}^k{\mathbb{P}(Y_i> y_i)}.\nonumber
\end{align}
\end{defi}

Next, in Prop. \ref{neg-placement}, we state the benefit of negatively correlated placement, which is the basis of future spatially correlated policies including our proposed policy in the current paper. 

\begin{prop} 
\label{neg-placement}
Negatively dependent content placement provides a higher average cache hit probability than the independent placement strategies.
\end{prop}

\begin{proof}
See Appendix \ref{App:Appendix-neg-placement}.
\end{proof}

%%%
In the remainder of this paper, we first discuss the independent content placement model in Sect. \ref{cachemodel-independent}, which is a special case of the geographic content placement ($\GCP$) problem using the Boolean model first proposed in \cite{Blaszczyszyn2014}. %and propose a linear approximation for the optimal cache placement probabilities. 

We then ask the following question: Given the coverage number $k$ and file $m$, how large cache hit rates can we achieve, i.e., how small can $\PmissN(m,k)\leq \mathbb{P}(Y_m=0)^k$ get for a spatial content placement setting, or what is the best negatively dependent content placement strategy? To answer that, we consider a negatively dependent content placement strategy inspired from the Mat\'{e}rn hard-core processes $\mhc$ (type II), which we call as the hard-core content placement ($\HCP$). We detail the $\HCP$ model in Sect. \ref{hardcore}.

\begin{table}[t!]\small%\footnotesize%\scriptsize
\begin{center}
\setlength{\extrarowheight}{3pt}
\begin{tabular}{l | c }
{\bf Symbol} & {\bf Definition} \\ 
\hline
{\bf General System Model Parameters} &\\
Baseline $\PPP$ with transmitter density $\tx$; a realization of the $\PPP$ & $\Phi$; $\phi=\{x_i\}\subset \mathbb{R}^2$ \\
$\DD$ communication radius; closed ball centered at $x_i$ with radius $\Rdd$ & $\Rdd$; $B_i(\Rdd)$\\
The coverage process of the $\DD$ transmitters driven by the $\BM$ & $V_{\rm BM}=\bigcup_i{(x_i+B_0(\Rdd))}$\\
File request distribution; Zipf request exponent & $p_r(\cdot)\sim \rm{Zipf}(\gamma_r)$; $\gamma_r$\\
Caching probability of file $m$ in cache $i$ & $\PX(m,x_i)$\\
Density of receivers; density of $\DD$ users & $\rx$; $\tx$\\
Number of $\DD$ users covering a receiver under strategy $\Pi$ & $\mathcal{N}_{\Pi}$\\
Hit probability for placement strategy $\Pi$ & $\PhitX$\\
Miss probability of file $m$ given $k$ users cover the \\
\hspace{0.3cm} typical receiver for placement strategy $\Pi$ & $\PmissX(m,k)$\\
Total number of files; cache size & $M; N<M$\\ 
\hline
{\bf Independent Content Placement Design} & \\
The caching distribution for independent placement & $\PI(m)$\\
The caching distribution for geographic content placement ($\GCP$) %strategy 
in \cite{Blaszczyszyn2014} & $\PG(m)$\\
The caching distribution for caching most popular content ($\MPC$)  & $\PP(m)=1_{m\leq N}$\\
%Cache design parameters for independent content placement & $L<N; K>N$\\
%The set of files that should be stored; discarded with probability 1 & $\{1,\hdots, L-1\}$; $\{K+1,\hdots, M\}$\\
\hline 
{\bf Hard-Core Content Placement ($\HCP$) Design} & \\
$\mhcA$ model constructed from the underlying $\PPP$ $\Phi$ & $\Phi_M$\\
Exclusion radius of file $m$ for the $\mhcA$ model & $r_m$\\
The density of the $\mhcA$ model for file $m$ & $\txMA(m)$\\
The number of neighboring transmitters in $B_0(r_m)$ & $C_m\sim\Poisson(\bar{C}_m)$, $\bar{C}_m=\tx\pi r_m^2$\\
The number of transmitters containing file $m$ in $B_0(\Rdd)$ & $\tilde{C}_m$\\
$2k$ dimensional bounded region $[0,D]^{2k}$ & $\mathcal{D}^k = [0,D]^{2k}$ \\
The cache miss region given there exists $k$ nodes & $\mathcal{V}^k= [0,D]^{2k}\backslash [0,\Rdd]^{2k}$\\
Second-order product density for file $m$ & $\rho^{(2)}_m(r)$ \\
\hline
\end{tabular}
\end{center}
\caption{Notation.}
\label{table:tab1}
\end{table}

%%%%%%%%%%%%%%%%%%%%%%%%%%%%%%%%%%%%%%%%%%%
\section{Independent Content Placement Design}
\label{cachemodel-independent}
Independent cache placement design is the baseline model where the files are cached at the $\DD$ users identically and independently of each other. Let $\PI(m)=p_c(m,x_i)=\mathbb{P}(Y_m=1)$ be the caching probability of file $m$ in any cache, which is the same at all points $x_i\in\phi$. 

The maximum average total cache hit probability, i.e., the probability that the typical user finds the content in one of the $\DD$ users it is covered by, can be evaluated by solving 
\begin{eqnarray}
\begin{aligned}
\label{eq:hitprob-opt}
\max_{\PI} &\,\,\, \PhitI\\
\textrm{s.t.}
& \quad \sum\limits_{m=1}^M{\PI(m)}\leq N,
\end{aligned}
\end{eqnarray}
and $\PmissI(m,k)=(1-\PI(m))^k$, which is related to $\PhitI$ through the $\PhitX$ expression in (\ref{PhitX}). 

First, we consider the following trivial case of independent placement, which is clearly suboptimal.
\begin{prop}\label{cachingmostpopular}
{\bf Caching most popular content $\MPC$.} 
The baseline solution is to store the most popular files only. Letting $Y_m=1_{m\leq N}$, i.e., $\PP(m)=1_{m\leq N}$, the miss probability is $\PmissMP(m,k)=1_{N<m\leq M}$ for all $m$ when $k\geq 1$, and $\PmissMP(m,k)=1$ when $k=0$. Hence, the average cache hit probability for the $\MPC$ scheme is $\PhitMP=\mathbb{P}(\mathcal{N}_{\rm MPC}\geq 1)\sum\nolimits_{m=1}^N{p_r(m)}$.
\end{prop}

%\subsection{Geographic Content Placement}
The independent cache design problem in our paper is a special case of the geographic content placement ($\GCP$) problem using the Boolean model as proposed in \cite{Blaszczyszyn2014}. The optimal solution of the $\GCP$ problem \cite{Blaszczyszyn2014} is characterized by Theorem \ref{GCP}.

\begin{theo}\label{GCP}
{\bf{Geographic Content Placement ($\GCP$) \cite[Theorem 1]{Blaszczyszyn2014}}}. The optimal caching distribution for the independent placement strategy is given as follows
\begin{eqnarray}
\label{pcoptimal}
\PGstar(m)=\begin{cases}
1,\quad \mu^*< p_r(m)\mathbb{P}(\mathcal{N}_P=1) \\
\frac{1}{\tx \pi {\Rdds}}\log\Big(\frac{p_r(m)\tx \pi {\Rdds}}{\mu^*}\Big),\quad p_r(m)\mathbb{P}(\mathcal{N}_P=1)\leq\mu^*\leq p_r(m)\mathbb{E}[\mathcal{N}_P]\\
0,\quad \mu^*> p_r(m)\mathbb{E}[\mathcal{N}_P]
\end{cases},
\end{eqnarray}
where $\mathbb{P}(\mathcal{N}_P=1)=e^{-\tx \pi {\Rdds}}(\tx \pi {\Rdds})$, $\mathbb{E}[\mathcal{N}_P]=\tx \pi\Rdds$. The placement probabilities satisfy 
\begin{align}
p_r(j)\sum\limits_{m=1}^M{\mathbb{P}(\mathcal{N}_P=m)m(1-\PGstar(j))^{m-1}}=\mu^*,\quad j\in\{1,\hdots,M\}. 
\end{align}
The optimal variable $\mu^*$ satisfies the equality $\sum\nolimits_{m=1}^M{\PGstar(m)}=N$.

Thus, the optimal value of the average cache hit probability for the $\GCP$ model is given by
\begin{align}
\label{IndependentHitProbability}
\PhitG=\sum\limits_{m=1}^M{p_r(m)[1-\exp{(-\tx \PGstar(m)\pi {\Rdds})}]}.
\end{align}
%, and for $p^*_{c}(m)\neq \{0,1\}$,
%\begin{align}
%p_r(m)\sum\limits_{k=1}^{\infty} k\mathbb{P}(\mathcal{N}_P=k)(1-p^*_{c}(m))^{k-1}=\mu^*
%\end{align} 
\end{theo}

\begin{proof}
See \cite[Theorem 1]{Blaszczyszyn2014}. It follows from the use of the Lagrangian relaxation method. The solution is found numerically using the bisection method.
\end{proof}

Throughout the paper we use the terms independent cache placement and $\GCP$ interchangeably.

\begin{comment}
%\subsection{A Linear Approximation to Independent Cache Design}
Given that each cache can store $N<M$ files\footnote{Swapping the contents within a cache does not change cache's state.}, our objective is to determine the number of files $L$ that should be stored in the cache with probability 1, and the maximum number of distinct files $K$ that can be stored as a function of the design parameters, e.g., $\Rdd$, $\tx$ and $N$. We uniquely determine $(L, K)$ that approximate the optimal content placement pmf in (\ref{pcoptimal}).

\begin{prop}\label{linapprox}
{\bf A linear approximation to $\GCP$.}
The following linear content placement model approximates (\ref{pcoptimal}):
\begin{eqnarray}
\label{pcapproximate}
\PGL(m)=\min\Big\{1,\Big(1-\frac{m-L}{K-L}\Big)^+\Big\},
\end{eqnarray}
where $y^+=\max\{y,0\}$ and the pair $(L, K)$ can be determined using (\ref{KL1}) and (\ref{Kopt}).
\end{prop}

\begin{proof}
See Appendix \ref{App:Appendix-linapprox}.
\end{proof}

We demonstrate that this linear model is a good approximation as shown in Sect. \ref{comparisonindependentmhc}.
\end{comment}

%%%%%%%%%%%%%%%%%%%%%%%%%%%%%%%%%%%%%%%%%%%
\section{Hard-Core Content Placement Design}
\label{hardcore}
We next consider the hard-core regime, which provides useful insights for the development of spatial content placement for the regime relevant to $\DD$ communications. Mat\'{e}rn's hard-core ($\mhc$) model is a spatial point process whose points are never closer to each other than some given distance. The proposed content placement model is slightly different from the $\mhc$ point process model with fixed radius. Instead, for each file type, a circular exclusion region is created around each $\DD$ transmitter such that the exclusion radius is determined by the popularity of the file. This is to prevent all the $\DD$ transmitters located in a circular region from caching a particular file simultaneously. %The exclusion radii are determined by the popularity distribution of the files. 

We provide two different spatially correlated content placement models both inspired from the Mat\'{e}rn hard-core ($\mhc$) (type II): (i) $\mhcA$ which is an optimized placement model to maximize the average total cache hit probability in (\ref{eq:hitprobX}), and (ii) $\mhcB$ which has the same marginal content placement probability as the $\GCP$ model in \cite{Blaszczyszyn2014}, and is sufficient for achieving a higher cache hit probability than the $\GCP$ model.

%%%%%
\subsection{Hard-Core Placement Model I ($\mhcA$)} 
\label{hardcore1}
We propose a content placement approach to pick a subset of transmitters based on some exclusion by exploiting the spatial properties of $\mhc$ (type II) model, which we call $\mhcA$. This type of $\mhc$ model is constructed from the underlying $\PPP$ $\Phi$ modeling the locations of the $\DD$ user caches by removing certain nodes of $\Phi$ depending on the positions of the neighboring nodes and additional marks attached to those nodes \cite[Ch. 2.1]{BaccelliBook1}. Each transmitter of the $\BM$ $V_{\rm BM}$ is assigned a uniformly (i.i.d.) distributed mark $U[0,1]$. A node $x\in\Phi$ is selected if it has the lowest mark among all the points in $B_x(R)$, given exclusion radius $R$. A realization of the $\mhc$ point process $\Phi_M$ is illustrated in Fig. \ref{Matern}. 

The $\mhcA$ placement model is motivated from the $\mhc$ model and implemented as follows. For each file type, there is a distinct exclusion radius ($r_m$ for file $m$) instead of having a fixed exclusion radius $R$. Given a realization $\phi$ of the underlying $\PPP$ modeling the locations of the transmitters with intensity $\tx$, we sort the file indices in order of decreasing popularity. For given file index $m$ and radius $r_m$, we implement the steps $(a)$-$(d)$ described in Fig. \ref{Matern} to determine the set of selected transmitters to place file $m$. For the same realization $\phi$, we implement this procedure for all files. Once a cache is selected $N$ times, then it is full, and no more file can be placed even if it is selected. The objective is to determine the file radii to optimize the placement.

\begin{comment}
We summarize the implementation of the $\mhcA$ model in Algorithm \ref{MHCAalgo}.
\begin{algorithm}[H]\small
 initialization\;
 \KwData{Given a realization of the underlying $\PPP$ with intensity $\tx$: $\phi$}
 \KwInput{$M$, $N$, $p_r(m)$, $r_m$}
 \KwResult{The hard-core content placement ($\HCP$) algorithm}
 \KwInitialization{Compute $\PMA^*(m)$, $\txMAstar$ and $r_m^*$ using (\ref{lambdamhcA}) and Theorem \ref{HCP}}
The number of available slots ($N_x$ for cache $x$) in all caches is $N$\;
%Cache overflow indicator\;
Cache underutilization indicator\;
 \For{$m=1:1:M$}{
 	\While{$|\Phi|>0$}{
  Associate a mark $\sim U[0,1]$ to each cache independently\;
  \eIf{$N_x> 0$}{
   Cache $x$ is selected if it has ``the lowest mark" in $B(x,r_m)$\;
     $N_x\leftarrow N_x-1$\;
   }{
   $\Phi\leftarrow \Phi\backslash \{x\}$ \;
  }
%  \If{a}{
% a	
%  } 
Node $x$ is selected if it has ``the lowest mark" in $B(x,r_m^*)$\;
 }
 }
 \caption{How to write algorithms}\label{MHCAalgo}
\end{algorithm}
\end{comment}

%%%%%%%%%%%%%%%%%%%%%%%%%%%%%%%%%%%%%%%%%%
\begin{defi}
{\bf Configuration probability.} 
The probability density function (pdf) of the $\mhc$ point process $\Phi_M$ with exactly $k$ points in a bounded region $\mathcal{D}=[0,D]^2\in\mathbb{R}^2$ that denotes the set retained caches that contain file $m$ is given by $f: \mathbb{R}^{2k}\to [0,\infty)$ \cite[Ch. 5.5]{Stoyan1996} so that
\begin{align}
\label{ffunctionMHC}
f_m(\varphi)=
\begin{cases}
a_m,\,\, \text{if} \,\, s_{\varphi}(r_m)=0,\\
0,\,\, \text{otherwise.}
\end{cases}
\end{align}
which is also known as the configuration probability, i.e., the probability that the hard-core model $\Phi_M$ takes the realization $\varphi$. In the above, $\varphi=\{x_1, \dots, x_k\}\subset \mathcal{D}$ denotes the set of $k$ points, $a_m$ is a normalizing constant and $s_{\varphi}(r)$ is the number of inter-point distances in $\varphi$ that are equal or less than $r$. This yields a uniform distribution\footnote{The pdf of the retained process (\ref{ffunctionMHC}) is a scaled version of the pdf of the $\PPP$ $\Phi$ in which there is no point within the exclusion range of the typical cache. This yields a uniform distribution of $k$ points in $\mathcal{D}$, i.e., $f(\varphi)=a$, where $a$ is a normalizing constant.} of a subset of $k$ points with inter-point distances at least $r_m$ in $\mathcal{D}$. 
\end{defi}
%%%%%

We optimize the exclusion radii to maximize the total hit probability. The exclusion radius of a particular file $r_m$ depends on the file popularity in the network, transmitter density and the cache size and satisfies $r_m<\Rdd$. Otherwise, once $r_m$ exceeds $\Rdd$, as holes would start to open up in the coverage for that content, the hit probability for file $m$ would suffer. We consider the following cases: (i) if the file is extremely popular, then many transmitters should simultaneously cache the file, yielding a small exclusion radius, and (ii) if the file is not very popular, then fewer (or zero) transmitters would be sufficient for caching the file, yielding a larger exclusion radius. Therefore, intuitively, we expect the exclusion radius to decrease with increasing file popularity. Our analysis also supports this conclusion that the exclusion radius is inversely related to the file popularity, i.e., the most popular files are stored in a high number of caches with higher marginal probabilities unlike the files with low popularity that are stored with lower marginals, with larger exclusion radius. 

By the Slivnyak Theorem, the Palm distribution of the $\PPP$ $\Phi$ seen from its typical point (cache) located at $0$ corresponds to the law of $\Phi \cup \{0\}$ under the original distribution \cite[Ch. 1.4]{BaccelliBook1}. Since the typical node (which is at the origin) of $\Phi$ has $C_m$ neighbors distributed as $C_m\sim\Poisson(\bar{C}_m)$ with $\bar{C}_m=\tx\pi r_m^2$, given the exclusion radius $r_m$ for file $m$ of the $\mhcA$ model, and the file may be placed at most at only one cache within this circular region. Hence, the probability of a typical $\DD$ transmitter to get the minimum mark in its neighborhood to qualify to cache file $m$, equivalently, the caching probability of file $m$ at a typical transmitter is
\begin{eqnarray}
\label{cacheprobmatern}
\PMA(m)=\mathbb{E}\Big[\frac{1}{1+C_m}\Big]=\frac{1-\exp(-\bar{C}_m)}{\bar{C}_m}.
\end{eqnarray}
From (\ref{cacheprobmatern}), we can easily observe that there is a one-to-one relationship between $r_m$ and $\PMA(m)$. The inverse relationship between $r_m$ and $\PMA(m)$ can be seen by taking the following limits:
\begin{align}
\label{rmvspcm}
\lim\limits_{r_m\to 0} \PMA(m)=1,\quad
\lim\limits_{r_m\to \infty} \PMA(m)=0,
\end{align}
which implies that the popular files have small $r_m$, hence are cached more frequently, and unpopular files have larger exclusion radii, and are stored at fewer locations.

We denote the density of the $\mhcA$ model for file $m$ by 
\begin{align}
\label{lambdamhcA}
\txMA(m)= \frac{[1 - \exp(-\bar{C}_m)]}{\pi r_m^2}  = \PMA(m)\tx. 
\end{align}
From (\ref{lambdamhcA}) and (\ref{cacheprobmatern}), we can see that the placement probability of file $m$ is the same as the percentage of nodes that cache the same file.

Let $\tilde{C}_m$ be the number of transmitters containing file $m$ within a circular region of radius $\Rdd$. At most one transmitter is allowed to contain a file within the exclusion radius. Therefore, when $r_m\geq\Rdd$, we have $\tilde{C}_m\in\{0,1\}$, and when $r_m<\Rdd$, we have $\tilde{C}_m\in \{0,1,2,\cdots\}$.

\begin{prop}
\label{MHCnegdep}
The $\mhc$ placement is a negatively dependent placement technique.
\end{prop}

\begin{proof}
See Appendix \ref{App:Appendix-MHCnegdep}.
\end{proof}

As the file popularity increases, the exclusion radius gets smaller. Hence, the average number of transmitters within the exclusion region, i.e., $\bar{C}_m^*$, decreases, and the chance of having at least one transmitter caching that file within $\Rdd$ increases, i.e., $\mathbb{P}(\tilde{C}_m\geq1)>\mathbb{P}(\tilde{C}_n\geq1)$ for $m<n$. This yields a higher $\PMA(\cdot)$ for more popular files from (\ref{cacheprobmatern}). If the demand distribution is uniform over the network, then each file has the same caching probability, i.e., $\PMA(m)$ is the same for all $m$, yielding the same $r_m$ for all $m$, which is intuitive. When the demand distribution is skewed towards the more popular files, then $\txMA(m)$ scales with the request popularity and $r_m$ is inversely proportional to $p_r(m)$, i.e., less popular files will end up being stored in fewer locations, and popular files will be guaranteed to be available over a larger geographic area, which is intuitive. 

In the $\mhcA$ model, using the pdf in (\ref{ffunctionMHC}) that denotes the configuration of the retained transmitters, the miss probability of file $m$ given $k$ users cover a typical receiver is
\begin{align}
\label{missMHC}
\PmissMA(m,k)=\idotsint\nolimits_{\mathcal{V}^k} f_m(x_1, \dots, x_k)   \,\mathrm{d}x_1\hdots \mathrm{d}x_k,
\end{align}
where the region $\mathcal{V}^k$ characterizes the cache miss region given there exists $k$ $\DD$ nodes, i.e., it is the $2k$ dimensional region denoted by $\mathcal{V}^k=[0,D]^{2k}\backslash [0,\Rdd]^{2k}$. 

The maximum hit probability for the $\mhcA$ model is given by the solution of  
\begin{equation}
\begin{aligned}
\label{eq:hitprob-matern}
\max_{\PMA(m)} &\quad \PhitMA\\
\textrm{s.t.}
& \quad \sum\limits_{m=1}^M{\PMA(m)}\leq N,
\end{aligned}
\end{equation}
and $\PmissMA(m,k)$ is given in (\ref{missMHC}), which is related to $\PhitMA$ through the $\PhitX$ expression given in (\ref{PhitX}) of the original optimization formulation in (\ref{eq:hitprobX}).

\begin{prop}\label{AvgHitProbmhcA}
The average cache hit probability for the $\mhcA$ model is
\begin{align}
\label{MHCaModelHitProbability}
\PhitMA=\sum\limits_{m=1}^M{p_r(m)\mathbb{P}(\tilde{C}_m> 0\vert r_m)},
\end{align} 
where the term $\mathbb{P}(\tilde{C}_m> 0\vert r_m)$ is essential in determining the cache hit probability and given as
\begin{align}
\label{distributionCm}
\mathbb{P}(\tilde{C}_m> 0\vert r_m)
\begin{cases}
\geq 1-\exp(-\txMA(m)\pi {\Rdds}), \quad r_m<\Rdd,\\
= \txMA(m)\pi\Rdds, \quad  r_m\geq\Rdd.
\end{cases}
\end{align}
\end{prop}

\begin{proof}
See Appendix \ref{App:Appendix-AvgHitProbmhcA}.
\end{proof}

\begin{figure*}[t!]
\centering
\includegraphics[width=\textwidth]{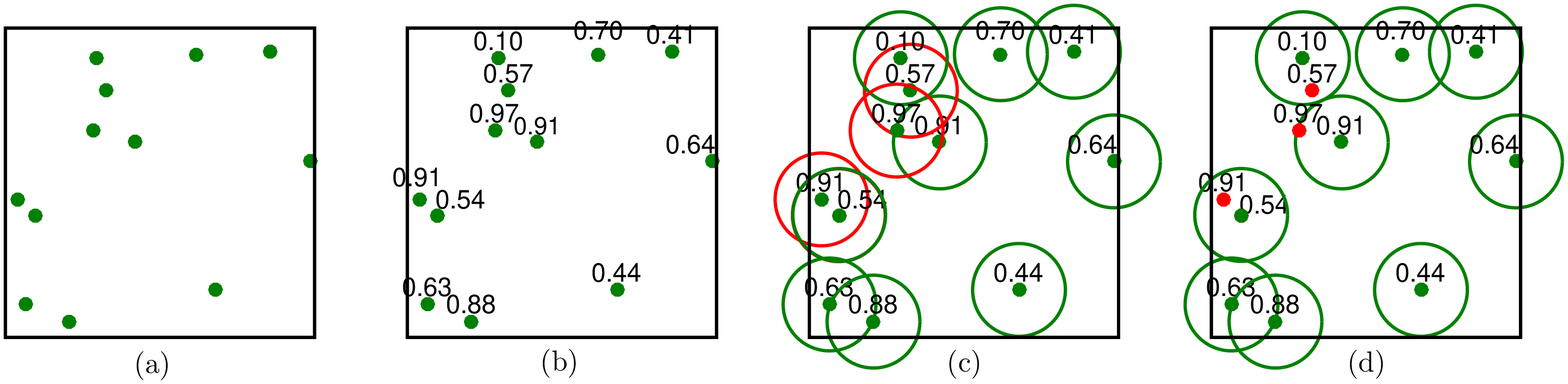}\\
\caption{\small{$\mhc$ point process realization for a given exclusion radius $R$: (a) Begin with a realization of $\PPP$, $\phi$. (b) Associate a uniformly distributed mark $U[0,1]$ to each point of $\phi$ independently. (c) A node $x\in\phi$ is selected if it has the lowest mark inside $B_x(R)$. (d) Set of selected points for a given realization of the $\PPP$. We exploit the $\mhc$ model to pick a subset of $\DD$ nodes to cache the files, where there is a distinct exclusion radius for each file, and the exclusion radii are determined by the underlying file popularity distribution.}}
\label{Matern}
\end{figure*}

The optimal solution of the $\mhcA$ model in (\ref{eq:hitprob-matern}) is characterized by Theorem \ref{HCP}.

\begin{theo}\label{HCP}
{\bf Hard-Core Content Placement ($\HCP$).} The optimal caching distribution for the $\HCP$ model is given as follows
\begin{align}
\label{pcoptforMHCA}
\PMAstar(m)=\begin{cases}
\tx^{-1}W(cp_r(m)),\quad m\leq \mc,\\
\tx^{-1}cp_r(m),\quad m>\mc,
\end{cases}
\end{align}
where $W$ is the Lambert  function, and $\mc=\argmax\limits_{m\in\{1,\cdots,M\}} \{r_m\vert r_m<\Rdd\}$, and the relation
\begin{align}
\label{lambdaoptforMatern}
\sum\limits_{m=1}^{\mc}{W(cp_r(m))-cp_r(m)}=N\tx-c
\end{align}
can be used to determined the value of $c$. Hence, we determine $\txMAstar(m)$ and the optimal value of the exclusion radius, i.e., $r^*_m$, from (\ref{lambdaoptforMatern}) as a function of the request pmf $p_r(m)$, cache size $N$ and the transmitter density $\tx$. 
\end{theo}

\begin{proof}
See Appendix \ref{App:Appendix-mhcAhitmaximum}.
\end{proof}

Consider a ball centered at origin and of radius $D$, i.e., $B_0(D)$, with $D\gg \max\nolimits_m\{r_m\}$, let the number of users in $B_0(D)$ be $\Poisson$ with $\mathbb{P}(\mathcal{N}_P(D)=k)=e^{-\bar{C}_D}\frac{(\bar{C}_D)^k}{k!}$, where $\bar{C}_D=\tx \pi D^2$ is the average number of transmitters within $B_0(D)$. Due to the limited storage capacity of the caches, the mean total number of files that can be cached in $B_0(D)$ is upper bounded by $N \bar{C}_D$. To determine the average number of users containing a desired file type in region $B_0(D)$, we use the second-order product density of the $\mhc$ process $\Phi_M$, which is defined next. 

\begin{defi}{\bf Second-order product density \cite[Ch. 5.4]{Stoyan1996}.}
For a stationary point process $\Phi_M$, the second-order product density is the joint probability that there are two points of $\Phi_M$ at locations $x$ and $y$ in the infinitesimal volumes $dx$ and $dy$% \cite{HaenggiBook2012}
, and given by %\cite{Haenggi2011} 
\begin{eqnarray}
\label{SOPD}
\rho_m^{(2)}(r)%=\lambda_t^2 k_m(r)\nonumber\\
=\begin{cases}
\txMAs(m),\quad r\geq 2r_m\\
\dfrac{2V_{r_m}(r)[1-\exp(-\tx \pi r_m^2)]-2\pi r_m^2[1-\exp(-\tx V_{r_m}(r))]}{\pi r_m^2V_{r_m}(r)[V_{r_m}(r)-\pi r_m^2]},\,\, r_m<r<2r_m,\\
0,\quad r\leq r_m
\end{cases}
\end{eqnarray}
where $\lambda_t^{-2}\rho_m^{(2)}(r)$ is the two-point Palm probability that two points of $\Phi$ separated by distance $r$ are both retained to store file $m$ \cite[Ch. 5.4]{Stoyan1996}, and $V_{r_m}(r)=2\pi r_m^2-2r_m^2\cos^{-1}\left(\frac{r}{2r_m}\right)+r\sqrt{r_m^2-\frac{r^2}{4}}$ is the area of the union of two circles with radius $r_m$ and separated by distance $r$. Pairwise correlations between the points separated by $r>r_m$ are modeled using the second-order product density --$\rho_m^{(2)}(r)$ for file $m$-- of the $\mhc$ process. 
\end{defi}

%{\bf Active Transmitters.} with second-order product density $\rho_m^{(2)}(r)$, 
Using the Campbell's theorem \cite[Ch. 1.4]{BaccelliBook1}, we deduce that the average number of transmitters of the stationary point process $\Phi_M$ --conditioned on there being a point at the origin but not counting it-- contained in the ball $B_0(\Rdd)$ is given by
\begin{eqnarray}
\label{AvgNoTx}
\mathbb{E}^{!\circ}\left[\sum\limits_{x\in\Phi_M} 1(x\in B_0(\Rdd))\right]={\tx}^{-1}\int\nolimits_{B_0(\Rdd)}\rho_m^{(2)}(x)\,dx.
\end{eqnarray}

An upper bound on the probability that a user requesting file $m$ is covered is given by the following expression:
\begin{align}
\label{probhavingoneTXsmallrmbound}
\mathbb{P}(\tilde{C}_m\geq 1\vert r_m<\Rdd)&\stackrel{(a)}{\leq} \mathbb{E}[\tilde{C}_m\vert r_m<\Rdd]\nonumber\\&\stackrel{(b)}{=}1-\exp(-\txMAstar(m) \pi {\Rdds})+{\tx}^{-1}\int\nolimits_{B_0(\Rdd)}\rho_m^{(2)}(x){\rm d}x,
\end{align}
where $(a)$ follows from using Markov inequality, and $(b)$ from using (\ref{AvgNoTx}), to deduce the average number of caches that stores file $m$ in $B_0(\Rdd)$.

\begin{prop}\label{cachehitprobMatern}
The maximum cache hit probability for the $\mhcA$ model is approximated by the following lower and upper bounds:
\begin{align}
\label{PhitMbounds}
\PhitMALB &= \sum\limits_{m=1}^{\mc}{p_r(m)[1-e^{-\txMAstar(m)\pi {\Rdds}}]}+\sum\limits_{m=\mc+1}^{M}{p_r(m)\txMAstar(m)\pi\Rdds},\nonumber\\
%\PhitMAUB &= \sum\limits_{m=1}^{\mc}{p_r(m)\Big[1-e^{-\txMAstar(m) \pi {\Rdds}}+{\tx}^{-1}\int\nolimits_{r_m^*}^{\Rdd}\rho_m^{(2)}(x){\rm d}x\Big]}\nonumber\\
%&+\sum\limits_{m=\mc+1}^{M}{p_r(m)\txMAstar(m)\pi\Rdds},\nonumber\\
\PhitMAUB &= \PhitMALB + \sum\limits_{m=1}^{\mc}{p_r(m) {\tx}^{-1}\int\nolimits_{r_m^*}^{\Rdd}\rho_m^{(2)}(x){\rm d}x},
\end{align}
where $\bar{C}_m^*=\tx\pi (r_m^*)^2$ with each $r_m^*$ denoting the optimal value of the radius $r_m$ for $m=1,\hdots, M$ that maximizes the average cache hit probability, and $\txMAstar(m)$ follows from plugging $r_m^*$ into (\ref{lambdamhcA}).
\end{prop}

\begin{proof}
See Appendix \ref{App:Appendix-cachehitprobMatern}.
\end{proof}

%%%%%%%%%%%%%%%%%%%%%%%%%%%%%%%%%%%%%%%%%%%%%%%%%%%%%%%%%%%%%%
To compare the performance of the $\GCP$ and the $\HCP$ models in terms of their average cache hit probabilities, we next consider an example.

\begin{ex}\label{ExGCPHCP}
{\bf Cache hit rate comparison for $\GCP$ and $\HCP$.} 
Consider a simple caching scenario with $M=2$ files and a cache size of $N=1$, and the request distribution satisfies $p_r(1)=2/3$ and $p_r(2)=1/3$. Let $\tx \pi=1$ and assume $\Rdd$ is given.
\begin{itemize}
%%%%%
\item In the $\GCP$ model, from Theorem \ref{GCP}, given the product $\tx \pi \Rdds$, the values of $\mathbb{P}(\mathcal{N}_P=1)$, $\mathbb{E}[\mathcal{N}_P]$ can be computed. Checking the conditions in (\ref{pcoptimal}), the optimal value of $\mu$, and $\PGstar(1)$ and $\PGstar(2)$ can be determined. Thus, from (\ref{IndependentHitProbability}), the optimal cache hit probability for the $\GCP$ model becomes $\PhitGstar=\sum\nolimits_{m=1}^2 p_r(m)[1-\exp(-\PGstar(m)\tx \pi \Rdds)]$.

%%%%%
\item In the $\HCP$ model, from (\ref{lambdamhcA}), we have $\txMA(m)= \frac{[1 - \exp(-\bar{C}_m)]}{\pi r_m^2}  = \PMA(m)\tx$ for $m=1,2$. Using the cache constraint, $\sum\nolimits_{m=1}^2\txMA(m)=\tx$. Thus, from (\ref{MHCaModelHitProbability}), the cache hit probability for the $\GCP$ model becomes $\PhitMA=2/3\mathbb{P}(\tilde{C}_1> 0\vert r_1)+1/3\mathbb{P}(\tilde{C}_2> 0\vert r_2)$, where from (\ref{distributionCm}), we compute $\mathbb{P}(\tilde{C}_m> 0\vert r_m)$ using the lower bound in Prop. \ref{cachehitprobMatern}. 

The optimal values $\PhitGstar$, $\PhitMALBstar$ for different $\Rdd$ are tabulated in Table \ref{table:tabExample1}. For $\Rdd$ high, as the lower bound of the $\HCP$ model is very close to $\PhitGstar$, both models perform similarly. However, for small $\Rdd$, the $\HCP$ model outperforms (with a cache hit rate gain up to $25\%$ using the lower bound) because it can exploit the spatial diversity.
\end{itemize}
\end{ex}

\begin{table}[t!]\small%\footnotesize%\scriptsize
\begin{center}
\setlength{\extrarowheight}{3pt}
\begin{tabular}{| l | c | c | c | c | c | c | c | c |}
\hline
$\Rdd$ & $\mu^*$ & $\PGstar(1)$, $\PGstar(2)$ & $\PhitGstar$ & $r_1^*$, $r_2^*$ & $\txMAstar(1)$, $\txMAstar(2)$  & $\PhitMALBstar$ \\ 
\hline
$\sqrt{0.5}$ & $0.1836$ & $1$, $0$ & ${\bf 0.2623}$ & $0.7071$, $1.7117$ & $0.2813$, $0.0370$ & ${\bf 0.3140}$ \\
$\sqrt{0.75}$ & $0.2430$ & $0.9621$, $0.0379$ & ${\bf 0.352}$ & $0.866$, $1.4283$ & $0.2428$, $0.0756$ & ${\bf 0.4407}$ \\
$1$ &  $.28592$  &$0.8466$, $0.1534$  &  ${\bf 0.4282}$  & $1$,  $1.257$ & $0.201$, $0.1174$ & ${\bf 0.5438}$ \\
$\sqrt{2}$ & $0.3468$ & $0.6733$, $0.3267$ & ${\bf 0.6532}$ & $0.8718, 1.4178$ & $0.2411, 0.0772$ & ${\bf 0.6818}$\\
$\sqrt{3}$ & $0.3156$ & $0.6155$, $0.3845$ &  ${\bf 0.7896}$  & $1.0149$,  $1.2410$      & $0.1961$, $0.1222$  &${\bf 0.7896}$ \\
$\sqrt{10}$& $0.0318$ & $0.5347$, $0.4653$ & ${\bf 0.9936}$ & $1.0909$, $1.1576$ & $0.1704$, $0.1479$ & ${\bf 0.9936}$\\
$10$ & $9.0926 e^{-21}$ & $0.5035$, $0.4965$ & ${\bf 1}$ & $1.1225$, $1.1225$ & $0.1592$, $0.1592$ & ${\bf 1}$\\
\hline
\end{tabular}
\end{center}
\caption{Numerical results for Example \ref{ExGCPHCP}, with $M=2$, $N=1$ and $p_r(1)=2/3$ $p_r(2)=1/3$, where the results for the $\HCP$ model are obtained by optimizing $\PhitMALB$ in (\ref{PhitMbounds}) of Proposition \ref{cachehitprobMatern}.}
\label{table:tabExample1}
\end{table}
%%%%%%%%%%%%%%%%%%%%%%%%%%%%%%%%%%%%%%%%%%%%%%%%%%%%%%%%%%%%%%

\begin{comment}
CAN YOU FIND A BOUND ON THE PROBABILITY THAT THE SAME CACHE IS SELECTED MORE THAN N TIMES. ANY SCALING RESULT BETWEEN N AND M?

What about cache overflow?
\end{comment}

Ideally, when a cache placement strategy is applied, the files need to be placed at a cache in a way that all the cache slots are occupied. In the $\GCP$ model in \cite{Blaszczyszyn2014}, authors propose a probabilistic placement policy to fill the caches. However, in the case of $\mhcA$ placement, due to the random assignment of the marks in each cache independently for distinct files, it is not guaranteed that all the caches are full in the $\mhcA$ approach, which causes underutilization of the caches as detailed next. 

\begin{prop}\label{mhcAunderutilization}
{\bf Cache underutilization.} 
The $\HCP$ placement model causes underutilization of the caches, i.e., on average, the fraction of the $\DD$ nodes of $\Phi$ that contain $N$ distinct files is always less than $1$. This can be formally stated as follows: 
\begin{align}
\label{cacheunderutilizationequation}
\frac{1}{N\mathbb{E}[\mathcal{N}_P]}\sum\limits_{m=1}^M{\mathbb{E}[\tilde{C}_m]}\leq 1,
\end{align}
where $\mathbb{E}[\mathcal{N}_P]=\tx \pi \Rdds$.
\end{prop}

\begin{proof}
See Appendix \ref{App:Appendix-mhcAunderutilization}.
\end{proof}

The storage size $N$ and the exclusion radius $r_m$ have an inverse relationship. As $N$ drops, because it is not possible to cache the files at all the transmitters, the exclusion radius should increase to bring more spatial diversity into the model. From the storage constraint in (\ref{eq:hitprob-matern}), as $N$ drops, $r_m$ increases ($r_m\to \infty$ as $N\to 0$). Hence, a typical receiver won't be able to find its requested files within its range. When $N$ increases sufficiently, $r_m$ can be made smaller so that more files can be cached at the same transmitter ($r_m\to 0$ as $N\to \infty$). Hence, the typical receiver will most likely have the requested files within its range.

\begin{prop}\label{mhcAsufficientcondition}
{\bf A sufficient condition for the $\mhcA$ placement model.} 
The $\mhcA$ performs better than the independent placement model ($\GCP$) \cite{Blaszczyszyn2014} in terms of hit probability if the following condition is satisfied:
\begin{align}
\label{lambdamhcAsufficient}
\txMA(m)\geq
\begin{cases}
\tx \PGstar(m),\quad r_m<\Rdd,\\
\dfrac{1-\exp(-\tx \PGstar(m)\pi {\Rdds})}{\pi\Rdds},\quad r_m\geq\Rdd,
\end{cases}
\end{align}
where $\PGstar(m)$ is the optimal caching distribution for the $\GCP$.
\end{prop}

\begin{proof}
See Appendix \ref{App:Appendix-mhcAsufficientcondition}.
\end{proof}

In the regime where $r_m$ is chosen to satisfy the inequality in (\ref{lambdamhcAsufficient}), for all $m$, the $\mhcA$ placement model performs better than independent placement, and the volume fraction occupied by the transmitters caching file $m$, i.e., the proportion of space covered by the union $\bigcup_{x_i\in\Phi_M}{(x_i+B_0(\Rdd))}$ pertaining to file $m$, is lower bounded by $\frac{\txMA(m)}{\tx} \geq \frac{1-e^{-\tx\PGstar(m) \pi \Rdds}}{\tx}$. When the selection of $\txMA(m)$ does not satisfy (\ref{lambdamhcAsufficient}), the volume fraction pertaining to the caches storing file $m$ is upper bounded by $\frac{\txMA(m)}{\tx}<\PGstar(m)$.

From (\ref{lambdamhcAsufficient}), the density parameter $\txMA(m)$ decreases with $\Rdd$, hence, the exclusion radius $r_m$ increases with $\Rdd$, which is intuitive because as the number of transmitters within the communication range increases a smaller fraction of them should cache the desired content. The exclusion radius decreases with popularity, i.e., $r_m$ decreases as $p_r(m)$ increases. It also decreases with $\tx$ and the cache size $N$. 

We consider two regimes of caching controlled by the cache size $N$, which determines the optimal cache placement solutions for the independent and $\mhcA$ placement models. The spatial diversity of the content is captured by the optimal placement distribution for given $N$. As $N$ increases, content diversity per cache increases and less spatial diversity is required. Therefore, when $N$ is sufficiently large, independent placement is better than $\mhcA$ placement. For the $\mhcA$ placement model, the exclusion radii decrease with the file popularity. However, for small $N$, a higher exclusion radii are required for all files, which will increase the spatial diversity. Therefore, in the regime where $N$ is small, for sufficiently large $\Rdd$, $\mhcA$ placement performs better than independent placement ($\GCP$).

We next detail another $\mhc$-based model called $\mhcB$ and provide sufficient conditions for achieving a higher cache hit probability than the $\GCP$ model of \cite{Blaszczyszyn2014}.

%%%%%%%%%%%%%%%%%%%%%
\subsection{Hard-Core Placement Model II ($\mhcB$)} 
\label{hardcore2}
In this section, we propose a new $\mhc$-inspired placement model called $\mhcB$. We seek a spatially correlated content caching model that improves the performance of the independent placement model of Sect. \ref{cachemodel-independent} based on the $\GCP$ problem in \cite{Blaszczyszyn2014} using the same marginal caching probabilities, i.e., on average the fraction of the users containing a file is equal to its optimal placement probability of the $\GCP$ model. 

Different from the $\mhcA$ model in Sect. \ref{hardcore1}, where we maximize the average cache hit probability given the finite cache storage constraint, in this section we optimize the exclusion radii using the caching distribution in (\ref{pcoptimal}) of the $\GCP$ model in Theorem \ref{GCP}, and provide sufficient conditions so that the $\mhcB$ model is at least as good as the $\GCP$ scheme of \cite{Blaszczyszyn2014}.

\begin{figure*}
\centering
\includegraphics[width=0.9\textwidth]{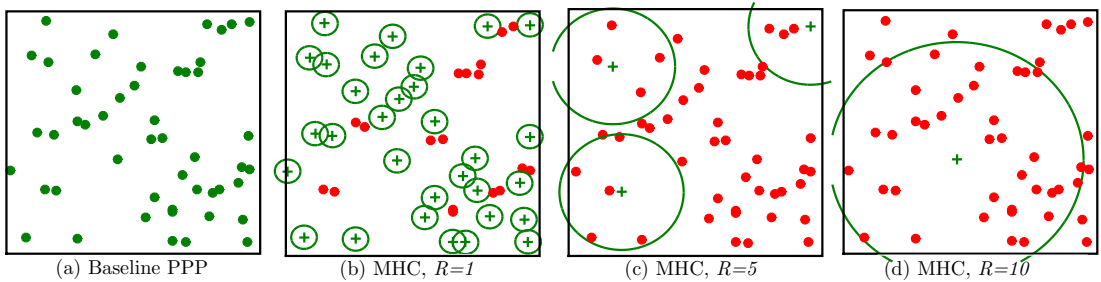}
\caption{\small{$\mhc$ versus the exclusion radii. Each node is associated a uniformly distributed mark $U[0,1]$ independently. Node $x_i\in\phi$ is selected if it has the lowest mark in $B_i(R)$. Selected nodes are denoted by plus sign. (a) Begin with a realization of $\PPP$, $\phi$. Set of selected points for a given realization of the $\PPP$ for an exclusion radius of (b) $R=1$, (c) $R=5$ and (d) $R=10$. As $R$ increases, the intensity of retained nodes decreases.}}
\label{mhc2vsexclusionradii}
\end{figure*}

The critical exclusion radius should be inversely proportional to the popularity of the requests, which is mainly determined by the skewness parameter $\gamma_r$. As $\gamma_r$ increases, the distribution becomes more skewed and higher variability is observed in the exclusion radii of different files.

In Fig. \ref{mhc2vsexclusionradii}, we illustrate the trend of the $\mhc$ process for different exclusion radii. As the exclusion radius $R$ increases, the intensity $\lambda_{\mhc}$ of $\mhcB$ process decreases.

\begin{prop}
\label{MHCAvsIndependent}
The exclusion radius for content $m$ for the $\mhcB$ model is given as 
\begin{align}
\label{mhcBexclusionradii}
r_m^B = \sqrt{\frac{1}{\lambda_t\pi}W\Big(-\frac{\exp(-1/\PGstar(m))}{\PGstar(m)}\Big)+\frac{1}{\lambda_t \pi \PGstar(m)}},\quad n\in\mathbb{Z},
\end{align}
where $\PGstar(\cdot)$ is the optimal caching distribution for $\GCP$ and $W$ is the Lambert  function.  
\end{prop}

\begin{proof}
See Appendix \ref{App:Appendix-MHCAvsIndependent}.
\end{proof}
%for m\leq \mc, we have r_m<\Rdd
From Prop. \ref{MHCAvsIndependent}, given the same marginal caching distributions for the $\GCP$ and the $\mhcA$ models, the relation (\ref{mhcBexclusionradii}) guarantees the $\mhcA$ model to outperform the independent content placement model in terms of the average cache hit rate performance.

Using the second order properties of the hard-core models, the variance of the $\HCP$ model %for a Borel set $B$, the variance per $\nu_2(B)$ --the intrinsic volume of $B$-- 
is approximated by 
%\begin{align}
$\rm{Var}_{\mhcA}%=\frac{\rm{Var}(\Phi_M(B))}{\nu_2(B)}
\simeq \txMA+2\pi \int\nolimits_{0}^{\infty}\left(\rho^{(2)}(r)-\txMAs\right)r{\rm d}r$ \cite[Ch. 4.5]{Stoyan1996}. 
%\end{align}
Hence, using (\ref{SOPD}) the variance of the $\mhc$ model for file $m$ can be approximated as
\begin{align}
\rm{Var}_{\mhcA}(m)\simeq %\txMA(m)-\pi\lambda^2_{\mhcA}(m)r_m^2+2\pi \lambda^2_{\mhcA}(m)\int\nolimits_{r_m}^{2r_m}\left(\frac{\rho_m^{(2)}(r)}{\txMAs(m)}-1\right)r{\rm d}r\nonumber\\ 
%&= \txMA(m)-4\pi \lambda^2_{\mhcA}(m)r_m^2+2\pi \int\nolimits_{r_m}^{2r_m}\rho_m^{(2)}(r)r{\rm d}r\nonumber\\
%&=
\txMA(m)-4\txMA(m)[1-\exp(-\tx \pi r_m^2)]+2\pi \int\nolimits_{r_m}^{2r_m}\rho_m^{(2)}(r)r{\rm d}r.
\end{align}

Note that $r_m$ decreases, and $\txMA(m)$ and $\rho_m^{(2)}(r)$ increase with popularity. Therefore, we can observe that there is a higher variability for popular files, which means that popular files are placed more randomly than unpopular files, and for unpopular files the placement distribution becomes more regular. This implies that randomized caching is in fact good for popular files, and more deterministic placement techniques are required for unpopular files.

%%%%%%%%%%%%%%%%%%%%%%%%%%%%%%%%%%%%%%%%%%%
\section{Numerical Comparison of Different Content Placement Models}
\label{comparisonindependentmhc}
We showed that the $\HCP$ techniques detailed in Sect. \ref{hardcore} yield negatively correlated placement, and can provide a higher cache hit than independent placement ($\GCP$). In this section, we verify our analytical expressions and provide a performance comparison between the $\GCP$ of \cite{Blaszczyszyn2014}, summarized in Sect. \ref{cachemodel-independent}, and the $\HCP$ of Sect. \ref{hardcore} by contrasting the average cache hit rates, as discussed in Sect. \ref{hitprobability}. For tractability, in our simulations we assume $M=2$ and $N=1$. The $\DD$ nodes form realizations of a PPP $\Phi$ over the region $[-10,10]^2$ with an intensity $\lambda_t$ per unit area. We assume there is a typical receiver at the origin which samples a request from the distribution satisfying $p_r(1)=2/3$ and $p_r(2)=1/3$. To compute the average cache hit probability performance of different models, we run $10^5$ iterations, where at each iteration, we consider a realization $\phi$ of PPP $\Phi$.

{\bf Cache hit rate with respect to $\lambda_t$.} 
We illustrate the cache hit probability trends of the $\MPC$ policy, the $\GCP$ model in \cite{Blaszczyszyn2014}, and the $\mhcA$ and $\mhcB$ placement models together with the bounds for the $\mhcA$ model with respect to the intensity $\lambda_t$ for $\Rdd=10$ in Fig. \ref{GCPvsHCPfordifferentlambdat}. It has already been numerically demonstrated in Fig. 3 of \cite{Blaszczyszyn2014} that the hit probability of $\GCP$ outperforms $\MPC$ policy, especially for low SINR thresholds, corresponding to large $\Rdd$ values. Therefore, we use $\GCP$ as benchmark for the comparison. The lower and upper bounds for the hit probability of the $\mhcA$ placement in (\ref{PhitMbounds}) of Prop. \ref{cachehitprobMatern} is also shown. Compared to the $\GCP$ model in \cite{Blaszczyszyn2014}, the $\mhcA$ and $\mhcB$ placement models provide higher cache hit probabilities, which we demonstrate next. From Fig. \ref{GCPvsHCPfordifferentlambdat}, we observe that the average cache hit probability for all cases improves with $\lambda_t$, $\GCP$ improves with increasing $\lambda_t$, and the performance gap between the $\HCP$ models and the $\GCP$ is higher at high $\lambda_t$. The respective cache hit gains of the $\mhcB$ and $\mhcA$ models over $\GCP$ can be up to $30\%$ and $37\%$, and the gain of $\mhcA$ over $\MPC$ is $50\%$ for this particular example. %Comment on the numerical gain for the simple model!
\begin{figure*}
\centering
\includegraphics[width=0.7\textwidth]{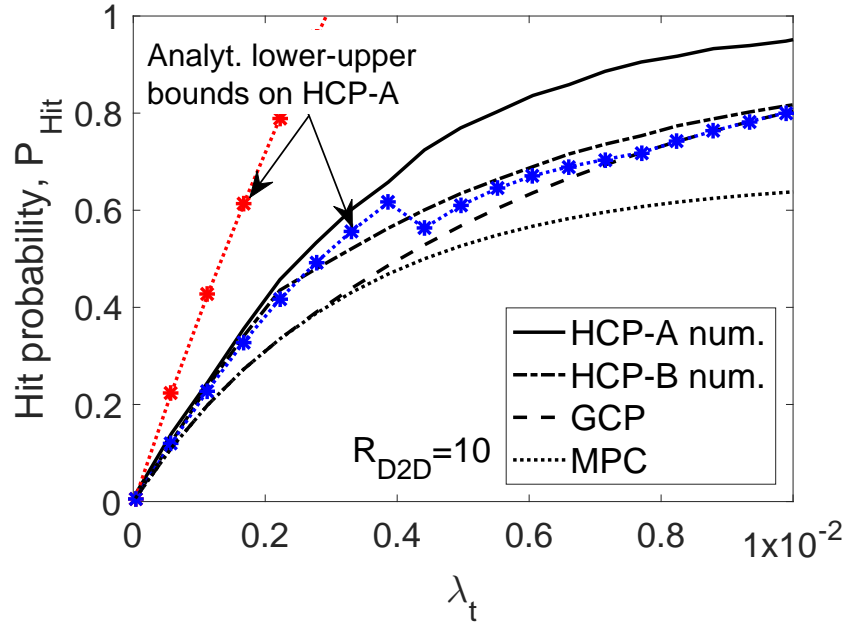}
\caption{\small{Maximum cache hit probabilities of the $\MPC$, $\GCP$ and $\HCP$ model for varying $\DD$ node intensity $\lambda_t$. }}
\label{GCPvsHCPfordifferentlambdat}
\end{figure*} 

{\bf Cache hit rate with respect to $\Rdd$.} 
The numerical comparison for the $\GCP$ and the $\mhcA$ models for varying $\Rdd$ and fixed $\lambda_t$ in Example \ref{ExGCPHCP} is tabulated in Table \ref{table:tabExample1}. Now, we illustrate the dependence of the average cache hit probability of different cache placement models on the communication radius $\Rdd$ in Fig. \ref{GCPvsHCPfordifferentRdd}. The lower and upper bounds for the hit probability of the $\mhcA$ placement in (\ref{PhitMbounds}) of Prop. \ref{cachehitprobMatern} is also shown. For high $\Rdd$, both models perform similarly. However, when $\Rdd$ is small, $\HCP$ performs better because it exploits the spatial diversity of the $\DD$ caches. For small $\Rdd$, feasible for the $\DD$ regime, $\mhc$-inspired approaches are a better alternative\footnote{One disadvantage of the $\mhcB$ model is that the excluded files' cache space is not reused, which can be resolved by jointly assigning marks. Therefore, we need to vectorize the marks to jointly determine the set of cached files and to avoid the problems caused by cache underutilization or overuse. The calculation of the cache underutilization or the overuse probability is left as future work.}.
\begin{figure*}
\centering
\includegraphics[width=0.7\textwidth]{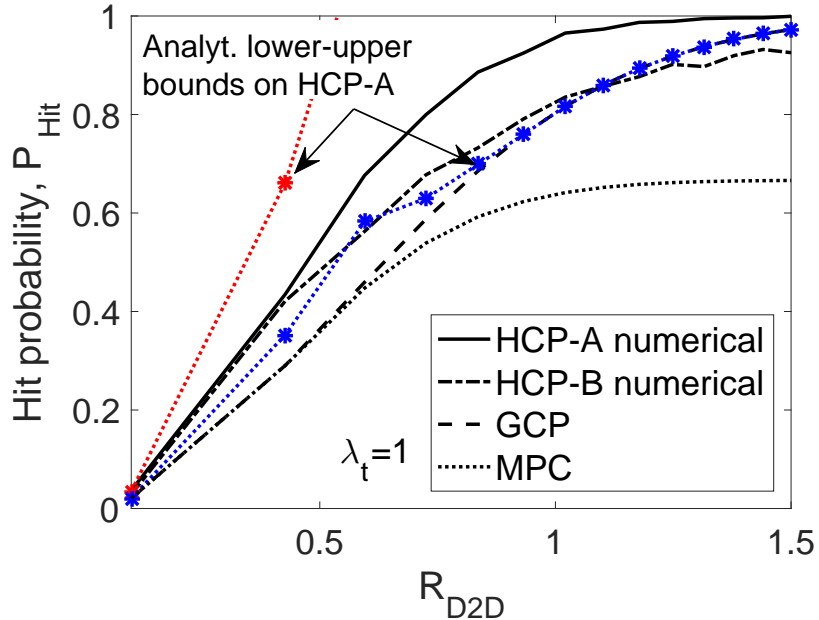}%lambdatOnebyPivsHitProbabilities
\caption{\small{Maximum cache hit probabilities of the $\MPC$, $\GCP$ and $\HCP$ models for varying communication radius. }}
\label{GCPvsHCPfordifferentRdd}
\end{figure*}

{\bf Cache utilization ratio.} As discussed in Proposition \ref{mhcAunderutilization}, the $\HCP$ placement model causes underutilization of the caches. We numerically investigate the cache utilization ratio for the $\mhcA$ sufficient condition given in Prop. \ref{mhcAsufficientcondition}, which is shown in Fig. \ref{UtilizationRatio}. As $\Rdd$ increases, the utilization drops because there will be more $\DD$ caches around the typical receiver and hence, the required number of cache slots decreases. For small $\lambda_t$, the values taken by $\txMA(m)$ are small that yields a low utilization ratio when $\Rdd$ is large, which follows from (\ref{lambdamhcAsufficient}). However, the utilization can be improved by jointly determining the values of $\txMA(m)$ and $\Rdd$.
\begin{figure*}[t!]
\centering
\center\includegraphics[width=0.6\textwidth]{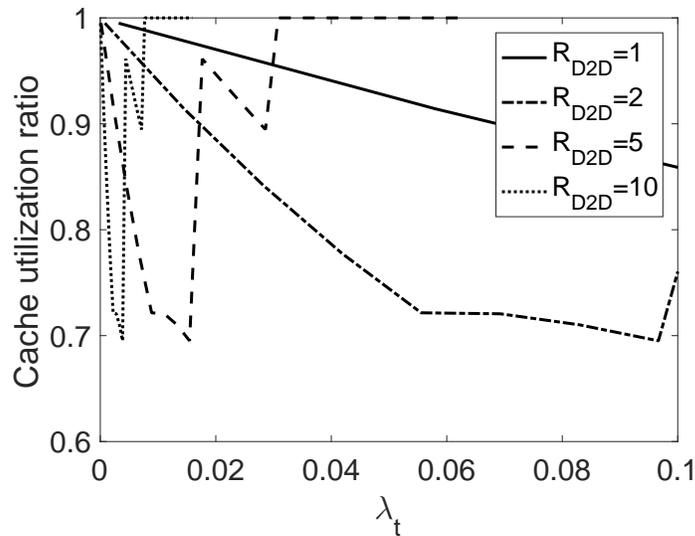}
\caption{\small{The cache underutilization (follows from the sufficient condition in Prop. \ref{mhcAsufficientcondition}).}\label{UtilizationRatio}}
\end{figure*}

{\bf Cache size.} The performance of the independent and the $\HCP$ models is mainly determined by the cache size. Hence, the analysis boils down to finding the critical cache size that determines which model outperforms the other in terms of the hit probability under or above the critical size. In Fig. \ref{exclusionradiiforMHCB}, we show the trend of the optimal exclusion radius $r_m$ of the $\mhcB$ model with respect to the caching pmf $\PX(m)$. As we expect from (\ref{lambdamhcAsufficient}), the exclusion radius $r_m$ decays with the popularity and the cache size $N$. Note that the $\HCP$ model compensates the small cache size at the cost of communication radius.
\begin{figure*}[t!]
\centering
\center\includegraphics[width=0.6\textwidth]{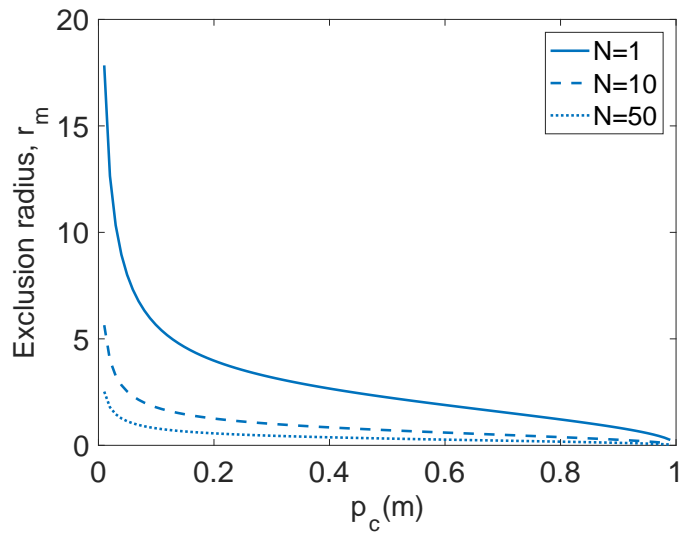}
\caption{\small{Characterization of the exclusion radii of $\mhcB$ for $N=[1, 10, 50]$ and %$M=100$, $\tx=.1$, 
$\Rdd=1$ as a function of $p_c(m)$.}\label{exclusionradiiforMHCB}}
\end{figure*}

{\bf Refinement to soft-core models.} 
The thinning leading to the $\mhc$ process can be refined such that higher intensities $\txMA$ are possible \cite[Ch. 5.4]{Stoyan1996}, at the price of more complicated algorithms \cite{Moller2010}  and \cite{Horig2012}. For refinement of the hard-core models, models based on Gibbs point processes ($\GPP$s) with repulsive potentials can be developed to generate soft-core\footnote{In the case of a soft-core point process, thinning is stronger the closer point pairs of the initial $\PPP$ are, but any pair distance still has non-vanishing probability.} placement models \cite[Ch. 18]{BaccelliBook1}. The study of soft-core models inspired from $\GPP$s, and the maximum caching gain due to the spreading of content in geographic settings is left as future work.

%%%%%%%%%%%%%%%%%%%%%%%%%%%%%%%%%%%%%%%%%%%
\section{Conclusions}
\label{conc} 
We proposed spatially correlated content caching models to maximize the hit probability by incorporating hard-core strategies that capture the pairwise interactions to enable spatial diversity.

{\em Our findings on spatial content caching suggest that the following design insights should enable more efficient caching models for $\DD$-enabled wireless networks:} 

{\bf Repulsive cache placement.} Negatively correlated content placement rather than independent placement is required to maximize the cache hit probability. Due to the isotropy of the $\PPP$ process, we contemplate a rotation invariant caching model. To satisfy negative spatial correlation, geographical separation of the content within the neighborhood of a typical receiver is required. Thus, in caching protocol design, it is important to incorporate an exclusion region around each cache, such that nodes in this region are not allowed to cache simultaneously. We show that high cache hit rates in a $\PPP$ network can be achieved through a $\mhc$-inspired placement model. 

{\bf Towards soft-core placement models.} We analyzed the $\HCP$ model, where the exclusions are determined by the hard-core radii. Future studies include more general solutions inspired from the $\GPP$ or Ising models capturing the pairwise interactions using soft-core potentials. The shape and scale of the potential %modeling the pairwise interactions 
should be determined accordingly. The pairwise potential function is promising because it can characterize the spatial and temporal dynamics of the file popularities at different geographic locations adaptively. Hence, the soft-core placement incorporating pairwise correlations can be exploited to improve the cache hit rate. This can can pave the way for the development of spatial cache placement and eviction policies to decide what content to discard, when to discard the content and where (to which neighbor) to relay the content, and provide practical design insights into how to adapt to geographical and temporal changes without compromising the accuracy. 

Possible extensions also include hierarchical models for content delivery \cite{Che2002}, multi-hop routing to improve the hit probability, distributed scheduling and content caching with bursty arrivals and delay constraints, and smoothing the cellular traffic by minimizing the peak-to-average traffic ratio with $\DD$ transmissions. 

%%%%%%%%%%%%%%%%%%%%%%%%%%%%%%%%%%%%%%%%%
\begin{appendix}
\section{Appendices}\label{Appendices}
%%%%%%%%%%%%%%%%%%%%%%%%%%%%%%%%%%%%%%%%%
\subsection{Proof of Proposition \ref{neg-placement}} 
\label{App:Appendix-neg-placement}
For a negatively dependent identical content placement, we can infer that $\PmissN(m,k)\stackrel{(a)}{\leq}\prod_{i=1}^k \mathbb{P}(Y_{(m,i)}=0)\stackrel{(b)}{=}\mathbb{P}(Y_m=0)^k$, where $(a)$ comes from Defn. \ref{NegativeDependence}, and $(b)$ is from identical content placement assumption. Hence, for a negatively dependent content placement strategy, the hit probability satisfies $\PhitN=1-\sum\nolimits_{m=1}^M{p_r(m)\sum\nolimits_{k=0}^{\infty}{\mathbb{P}(\mathcal{N}_P=k)\PmissN(m,k)}}\overset{(a)}{\geq} 1-\sum\nolimits_{m=1}^M{p_r(m)\sum\nolimits_{k=0}^{\infty}{\mathbb{P}(\mathcal{N}_P=k)\mathbb{P}(Y_m=0)^k}}$, where the $\RHS$ of $(a)$ is the hit probability for independent placement for $\PI(m)=1-\mathbb{P}(Y_m=0)$.

\subsection{Proof of Proposition \ref{MHCnegdep}}
\label{App:Appendix-MHCnegdep}
Dropping the file index $m$, let $Y_{i}$ be the indicator random variable that takes the value $1$ if file $m$ is available in the cache located at $x_i\in\phi$ and $0$ otherwise. Given the typical node has $k$ neighbors within its exclusion radius, from (\ref{lambdamhcA}), the probability that a node $x\in\Phi$ is selected, i.e., has the lowest mark among all the points in $B_x(r)$, to cache the file is $\mathbb{P}(Y_i=1)=1/(k+1)$. For $k>1$, $\mathbb{P}\Big(\bigcap\nolimits_{i=1}^k {Y_i=0}\Big)=\mathbb{P}\Big(\bigcap\nolimits_{i=1}^k {Y_i=1}\Big)=0$ since the probability that all nodes are assigned the same mark values is $0$. Therefore, the following relations in Definition \ref{NegativeDependence} hold:
\begin{align}
\mathbb{P}\Big(\bigcap\nolimits_{i=1}^k {Y_i=0}\Big) < \prod\nolimits_{i=1}^k{\mathbb{P}(Y_i= 0)},\quad
\mathbb{P}\Big(\bigcap\nolimits_{i=1}^k {Y_i=1}\Big) <\prod\nolimits_{i=1}^k{\mathbb{P}(Y_i=1)},\nonumber
\end{align}
and the $\mhc$ placement satisfies the negative dependence condition in Proposition \ref{neg-placement}.

%%%%%%%%%%%%%%%%%%%%%%%%%%%%%%%%%%%%%%%%%
\subsection{Proof of Proposition \ref{AvgHitProbmhcA}}
\label{App:Appendix-AvgHitProbmhcA}
%Case 1
We first consider the case $r_m\geq\Rdd$, where the user can be covered by at most one transmitter that has file $m$. The probability that the user is covered is given by the probability that there exists a transmitter of the $\mhcA$ process of file $m$ at the origin as determined by \cite[Ch. 2.1]{BaccelliBook1}  
\begin{eqnarray}
\label{probhavingoneTX}
\mathbb{P}(\tilde{C}_m=1\vert r_m\geq\Rdd)= \mathbb{E}[\tilde{C}_m \vert r_m\geq\Rdd] %\nonumber\\
= \txMA(m)\pi \Rdds = [1-e^{-\bar{C}_m}] \Big(\frac{\Rdd}{r_m}\Big)^2.
\end{eqnarray} 
% Case 2
For the case where $r_m<\Rdd$, we can estimate $\mathbb{P}(\tilde{C}_m\geq 1\vert r_m<\Rdd)$ using the second-order product density of the MHC model. However, we use a simpler approximation for tractability. The probability that a transmitter is eliminated in the $\mhcA$ with exclusion radius $r_m$ is equal to $1-\frac{\txMA(m)}{\tx}$. For the case of $r_m<\Rdd$, let the number of points in $B(r_m)$ from the original $\PPP$ satisfy $\Phi(B_0(\Rdd))=k$. Since $\mhcA$ is negatively correlated, from Definition \ref{NegativeDependence}, we can exploit the $\PPP$ approximation for the $\mhc$ in \cite{Ibrahim2013} to calculate the following upper bound for the probability that $k$ points are eliminated in $\mhcA$:
\begin{eqnarray}
\label{kpointseliminate}
\mathbb{P}(k\,\,\text{points are eliminated in $\mhcA$ $\Phi_M$ with}\, r_m\vert \mathcal{N}_P=k)\leq\Big(1-\frac{\txMA(m)}{\tx}\Big)^k.
\end{eqnarray}
Using (\ref{kpointseliminate}), the void probability of the $\mhcA$ is approximated as
\begin{eqnarray}
\label{probhavingoneTXsmallrm}
\mathbb{P}(\tilde{C}_m=0\vert r_m<\Rdd)
\leq\sum\limits_{k=0}^{\infty}\mathbb{P}(\mathcal{N}_P=k)\Big(1-\frac{\txMA(m)}{\tx}\Big)^k%\nonumber\\
= e^{-\txMA(m)\pi {\Rdds}}.
\end{eqnarray}
The relations (\ref{probhavingoneTX}) and (\ref{probhavingoneTXsmallrm}) yield the final result.

%%%%%%%%%%%%%%%%%%%%%%%%%%%%%%%%%%%%%%%%%
\subsection{Proof of Theorem \ref{HCP}} 
\label{App:Appendix-mhcAhitmaximum}
Define the Lagrangian to find the solution of (\ref{eq:hitprob-matern})  as follows: 
\begin{multline}
\mathcal{M}(\zeta)=\sum\limits_{m=1}^M{p_r(m)\mathbb{P}(\tilde{C}_m>0\vert r_m)}+\zeta \Big(\sum\limits_{m=1}^M{\frac{\txMA(m)}{\tx}-N}\Big)\nonumber\\
\stackrel{(a)}{\approx} \sum\limits_{m=1}^{\mc}{p_r(m)\big[1-e^{-\txMA(m)\pi {\Rdds}}\big]}+\sum\limits_{m=\mc+1}^{M}{p_r(m)\txMA(m)\pi\Rdds}+\zeta \Big(\sum\limits_{m=1}^M{\frac{\txMA(m)}{\tx}}-N\Big),\nonumber
\end{multline}
where $\mc=\argmax\nolimits_{m\in\{1,\cdots,M\}} \{r_m\vert r_m<\Rdd\}$, and $(a)$ follows from the void probability of the $\mhcA$ for $r_m<\Rdd$ given in (\ref{probhavingoneTXsmallrm}), and the probability that the user is covered for $r_m\geq\Rdd$ as given in (\ref{probhavingoneTX}). Taking its derivative with respect to $\txMA(m)$,
\begin{align}
\frac{d \mathcal{M}(\zeta)}{d \txMA(m)}=
\begin{cases}
p_r(m)\pi\Rdds e^{-\txMA(m)\pi {\Rdds}}+\zeta\frac{\txMA(m)}{\tx},\quad m\leq \mc \\
p_r(m)\pi\Rdds+\zeta\frac{\txMA(m)}{\tx},\quad m>\mc
\end{cases}\nonumber
\end{align}
Evaluating this at $\left.\frac{d \mathcal{M}(\zeta)}{d \txMA(m)}\right\vert _{\zeta=\zeta^*} =0$, we obtain 
\begin{align}
\zeta^*=\begin{cases}
-p_r(m)\frac{\tx\pi\Rdds}{\txMA(m)}e^{-\txMA(m)\pi {\Rdds}},\quad m\leq \mc \\
-p_r(m)\frac{\tx\pi\Rdds}{\txMA(m)},\quad m>\mc
\end{cases}.\nonumber
\end{align}
Note that the optimal solution $\zeta^*$ is increasing in the optimal value of $\txMA(m)$, i.e., $\txMAstar(m)$. To satisfy this relation, $\txMAstar(m)$ has to satisfy $\txMAstar(m)e^{\txMAstar(m)}= cp_r(m)$ for $m\leq \mc$ and $\txMAstar(m)= cp_r(m)$ for $m>\mc$ for a constant $c$. Using the constraint in (\ref{eq:hitprob-matern}), we obtain the relation (\ref{lambdaoptforMatern}) that determines the value of $c$.

%%%%%%%%%%%%%%%%%%%%%%%%%%%%%%%%%%%%%%%%%
\subsection{Proof of Proposition \ref{cachehitprobMatern}}
\label{App:Appendix-cachehitprobMatern}
Incorporating the pdf of the $\mhc$ point process with exactly $k$ points given in (\ref{ffunctionMHC}) into the miss probability of the $\mhcA$ model in (\ref{missMHC}), we derive the cache miss probability for the $\mhcA$ model, i.e., the probability that $k$ caches cover a receiver, and none has file $m$, as follows
\begin{align}
\PmissMA(m,k)%&=\idotsint\nolimits_{\mathcal{V}^k} f_m(x_1, \dots, x_k)   \,\mathrm{d}x_1 \dots \mathrm{d}x_k\nonumber\\
%&
\begin{cases}
\leq {\big(\int\nolimits_{r_m}^D{{\rm d}x}\big)^k}\Big/{\big(\int\nolimits_{0}^D{{\rm d}x}\big)^k}\stackrel{(a)}{=}\Big(1-\frac{r_m^2}{D^2}\Big)^k,\quad r_m<\Rdd\\
=\dfrac{\int\nolimits_0^1{(1-\exp{(-\bar{C}_m t})){\rm d}t}}{\bar{C}_m^{-1}}\big(\frac{\Rdd}{r_m}\big)^2=\exp(-\bar{C}_m)\big(\frac{\Rdd}{r_m}\big)^2,\quad r_m\geq \Rdd
\end{cases},\nonumber
\end{align}
where $\mathcal{V}^k$ characterizes the cache miss region given $k$ nodes, and $(a)$ follows from converting the integral into polar coordinates.

Since $\PmissMA(m,k)$ is related to $\PhitMA$ through the $\PhitX$ expression given in (\ref{PhitX}), the lower bound $\PhitMALB$ on the maximum cache hit probability for the $\mhcA$ model is given by the final expression in (\ref{PhitMbounds}), which gives the solution of the the $\mhcA$ hit probability maximization formulation in (\ref{eq:hitprob-matern}) given the $r_m$ values are optimized using the relation (\ref{lambdaoptforMatern}) of Theorem \ref{HCP}. Similarly, the upper bound $\PhitMAUB$ can be found using (\ref{probhavingoneTXsmallrmbound}). 

%%%%%%%%%%%%%%%%%%%%%%%%%%%%%%%%%%%%%%%%%
\subsection{Proof of Proposition \ref{mhcAunderutilization}} 
\label{App:Appendix-mhcAunderutilization}
In the $\mhcA$ model with exclusion radius $r_m$, from (\ref{lambdamhcA}), the average number of nodes in $B_0(r_m)$ is given by $1-e^{-\bar{C}_m}$. For popular files with $r_m<\Rdd$, the maximum number of non-overlapping circles with radius $r_m$ that can fit inside $B_0(\Rdd)$ is upper bounded by $\Big(\frac{\Rdd}{r_m}\Big)^2$. Hence, the following inequality is satisfied:
\begin{align}
\label{sufficientMHC}
\mathbb{E}[\tilde{C}_m\vert r_m<\Rdd]%&=1-\exp(-\txMA(m) \pi {\Rdds})+2\pi{\tx}^{-1}\int\nolimits_{r_m}^{\Rdd}\rho_m^{(2)}(r)r{\rm d}r \nonumber\\
%&\leq \pi{\tx}^{-1}\txMAs(m)(\Rdds-r_m^2) BU SATIR YANLIS\nonumber\\
%&= (\pi{\tx}^{-1})\Big({\tx}\frac{[1-e^{-\bar{C}_m}]}{\bar{C}_m}\Big)^2 r_m^2\Big(\Big(\frac{\Rdd}{r_m}\Big)^2-1\Big)\nonumber\\
%&=\frac{[1-e^{-\bar{C}_m}]^2}{\bar{C}_m}\Big(\Big(\frac{\Rdd}{r_m}\Big)^2-1\Big)\nonumber\\
\leq [1-e^{-\bar{C}_m}]\Big(\frac{\Rdd}{r_m}\Big)^2.
\end{align}
%where the last step follows from $1-e^{-x}\leq x$ for $x\geq 0$.

The average number of transmitters that cache all the files in $B_0(\Rdd)$ is given by
\begin{align}
\sum\limits_{m=1}^M{\mathbb{E}[\tilde{C}_m]}&=\sum\limits_{m=1}^{\mc} \mathbb{E}[\tilde{C}_m\vert r_m<\Rdd]+\sum\limits_{m=\mc+1}^{M} \mathbb{E}[\tilde{C}_m\vert r_m\geq\Rdd] \nonumber\\
%&=\sum\limits_{m=1}^{\mc} 2\pi{\tx}^{-1}\int\nolimits_{r_m}^{\Rdd}\rho_m^{(2)}(r)r{\rm d}r + \sum\limits_{m=\mc+1}^{M}\txMA(m)\pi \Rdds\nonumber\\
%&\leq \sum\limits_{m=1}^{\mc} \pi{\tx}^{-1}\txMAs(m)(\Rdds-r_m^2) + \sum\limits_{m=\mc+1}^{M} \txMA(m)\pi\Rdds \nonumber
&\stackrel{(a)}{\leq} \sum\limits_{m=1}^{\mc} [1-\exp(-\bar{C}_m)] \Big(\frac{\Rdd}{r_m}\Big)^2 + \sum\limits_{m=\mc+1}^{M}\txMA(m)\pi \Rdds\nonumber\\
&=\sum\limits_{m=1}^M{[1-\exp(-\bar{C}_m)] \Big(\frac{\Rdd}{r_m}\Big)^2}\stackrel{(b)}{\leq} N\mathbb{E}[\mathcal{N}_P]=N\tx \pi\Rdds,\nonumber
\end{align}
which results in underutilization of the caches. In the above, the inequality in $(a)$ follows from the inequality (\ref{sufficientMHC}) for $r_m<\Rdd$, and the equality (\ref{probhavingoneTX}) for $r_m\geq \Rdd$, and $(b)$ from scaling the constraint $\sum\nolimits_{m=1}^M{\PMA(m)}\leq N$ of the hit probability maximization formulation in (\ref{eq:hitprob-matern}) with $\tx \pi \Rdds$ and using the relation (\ref{cacheprobmatern}). Thus, the main reason for the underutilization is the popular files with $r_m<\Rdd$. Despite the underutilization of the caches, from (\ref{kpointseliminate}), the void probability of the very popular files will be insignificant. As the skewness of the Zipf distribution increases, $r_m$ for popular $m$ becomes even smaller and in the limit as $p_r(1)$ goes to $1$, the value of $r_1$ converges to $0$. In that case, since the probability of jointly retaining two nodes separated by any distance will be independent of each other, we observe that (\ref{kpointseliminate}) will be satisfied with equality. Therefore, the inequality in $(a)$ above that causes the underutilization in (\ref{cacheunderutilizationequation}) will eventually become equality.

%%%%%%%%%%%%%%%%%%%%%%%%%%%%%%%%%%%%%%%%%
\subsection{Proof of Proposition \ref{mhcAsufficientcondition}} 
\label{App:Appendix-mhcAsufficientcondition}
Using the hit probabilities given in (\ref{eq:hitprob-opt}) and (\ref{eq:hitprob-matern}), respectively for the independent and $\mhcA$ content placements, a necessary condition for the $\mhcA$ to perform better than the optimal independent placement model in \cite{Blaszczyszyn2014} in terms of hit probability is given by
\begin{eqnarray}
\label{mhcbsufficient}
\PhitMA=\sum\limits_{m=1}^M{p_r(m)\mathbb{P}(\tilde{C}_m\geq 1\vert r_m)} \geq \PhitG=\sum\limits_{m=1}^M{p_r(m)[1-\exp(-\tx \PGstar(m)\pi {\Rdds})]}.
\end{eqnarray}
A sufficient condition for (\ref{mhcbsufficient}) to be valid is $\mathbb{P}(\tilde{C}_m\geq1\vert r_m)\geq 1-\exp{(-\tx \PGstar(m) \pi {\Rdds})}$. For files with very high popularity, from (\ref{probhavingoneTXsmallrm}) we have
\begin{eqnarray}
\label{cond1}
\mathbb{P}(\tilde{C}_m\geq 1\vert r_m<\Rdd)\geq 1-\exp(-\txMA(m)\pi {\Rdds})\geq 1-\exp(-\tx \PGstar(m)\pi {\Rdds}).
\end{eqnarray}
For files with very low popularity, $r_m$ tends to be very high, i.e., $r_m\geq\Rdd$, and from (\ref{probhavingoneTX}), 
\begin{align}
\label{cond2}
\mathbb{P}(\tilde{C}_m=1\vert r_m\geq\Rdd)=\txMA(m)\pi \Rdds \geq 1-\exp(-\tx \PGstar(m)\pi {\Rdds}). 
\end{align}
Solving (\ref{cond1}) and (\ref{cond2}), the final result is obtained. 

The following relation is established from (\ref{cond1}) and (\ref{cond2}):
\begin{align}
\label{cond3}
\sum\limits_{m=1}^M{\txMA(m)}&\geq \sum\limits_{m=1}^{\mc}{\tx \PGstar(m)}+\sum\limits_{m=\mc+1}^{M}{\frac{1-\exp(-\tx \PGstar(m)\pi {\Rdds})}{\pi\Rdds}},
\end{align} 
where using $1-e^{-x}\leq x$ for $x\geq 0$, the RHS of (\ref{cond3}) can be shown to satisfy:
\begin{align}
\leq\sum\limits_{m=1}^{\mc}{\tx \PGstar(m)}+\sum\limits_{m=\mc+1}^{M}{\tx \PGstar(m)}=\tx\sum\limits_{m=1}^{M}{\PGstar(m)}=N\tx.\nonumber
\end{align}
For a feasible cache placement strategy, we also require that $\sum\nolimits_{m=1}^M{\txMA(m)}\leq N\tx$. Hence, it is possible to set $\txMA(m)$'s as in (\ref{lambdamhcAsufficient}) and satisfy the feasible placement condition.

\begin{comment}
Subtracting the respective hit probabilities, we obtain
\begin{align}
\PhitMA-\PhitI&=\sum\limits_{m=1}^{\mc}{p_r(m)[1-e^{-\txMA(m)\pi {\Rdds}}]}+\sum\limits_{m=\mc+1}^{M}{p_r(m)\txMA(m)\pi {\Rdds}}\nonumber\\
&-\sum\limits_{m=1}^{M}{p_r(m)[1-e^{-\tx \PGstar(m)\pi {\Rdds}}]}\nonumber\nonumber\\
&=\sum\limits_{m=1}^{\mc}{p_r(m)\Big[e^{-\tx \PGstar(m)\pi {\Rdds}}-e^{-\txMA(m)\pi {\Rdds}}\Big]}\nonumber\\
&+\sum\limits_{m=\mc+1}^{M}{p_r(m)\Big[e^{-\tx \PGstar(m)\pi {\Rdds}}-1+\txMA(m)\pi {\Rdds}\Big]}\nonumber\\
&>\sum\limits_{m=1}^{\mc}{p_r(m)[1-e^{-(1-e^{-\tx \pi {\Rdds}})}]}
\end{align}
\end{comment}

%%%%%%%%%%%%%%%%%%%%%%%%%%%%%%%%%%%%%%%%%
\subsection{Proof of Proposition \ref{MHCAvsIndependent}} 
\label{App:Appendix-MHCAvsIndependent}
In order to compute the exclusion radii $r_m^B$ for the $\mhcB$ model, we relate the expression (\ref{cacheprobmatern}) for the marginal caching probability of file $m$ of the $\mhcA$ model, which is also true for the $\mhcB$ model, to the optimal placement probability for the $\GCP$ model in \cite{Blaszczyszyn2014} such that the solution for the exclusion radius $r_m^B=\sqrt{\bar{C}_m/(\tx\pi)}$ satisfies
\begin{align}
\PGstar(m)=\frac{\txMA(m)}{\tx}=\frac{1-\exp(-\bar{C}_m)}{\bar{C}_m}.
\end{align}

For such selection of variables $r_m^B$'s, the hit probability for the $\mhcB$ model satisfies
\begin{align}
\label{PhitMmhcB}
\PhitMB&=\sum\limits_{m=1}^{\mc}{p_r(m)[1-e^{-\txMA(m)\pi {\Rdds}}]}+\sum\limits_{m=\mc+1}^{M}{p_r(m)[1-\exp(-\bar{C}_m)]\Big(\frac{\Rdd}{r_m^B}\Big)^2}\nonumber\\
&=\sum\limits_{m=1}^{\mc}{p_r(m)[1-e^{-\tx \PGstar(m) \pi {\Rdds}}]}+\sum\limits_{m=\mc+1}^{M}{p_r(m)\tx \PGstar(m)\pi {\Rdds}}.
\end{align}

On the other hand, the hit probability for the $\GCP$ model in \cite{Blaszczyszyn2014} satisfies (\ref{mhcbsufficient}). Noting that $x\geq 1-e^{-x}$ for $x\geq 0$, hence from (\ref{IndependentHitProbability}) and (\ref{PhitMmhcB}), we conclude that $\PhitMB\geq \PhitG$.

\end{appendix}

%\section*{Acknowledgement}
%Authors thank the reviewers for their critical feedback and the significant improvement of the paper.

\begin{spacing}{1.32}%1.465
\bibliographystyle{IEEEtran}
\bibliography{D2Dreferences}
\end{spacing}

\end{document}